\def\fullversionflag{1}

\ifodd\fullversionflag
  \documentclass[11pt]{article}
\else
  \documentclass[conference]{IEEEtran}
\fi

\usepackage{etoolbox}
\newtoggle{fullversion}
\ifodd\fullversionflag
  \toggletrue{fullversion}  
\else
  \togglefalse{fullversion}
\fi

\usepackage{amsmath}
\usepackage{amsthm}
\usepackage{amsfonts}
\usepackage{amssymb}
\usepackage{graphpap}
\usepackage{graphics}

\usepackage[normalem]{ulem}
\usepackage{array}
\usepackage{enumerate}
\usepackage{framed}
\iftoggle{fullversion}{}{
  \let\labelindent\relax
}
\usepackage{enumitem}
\usepackage{float}
\usepackage{mathrsfs}
\usepackage{etoolbox}
\usepackage{graphicx}
\iftoggle{fullversion}{
  \usepackage[page]{appendix}
}{}
\usepackage[pdfborder={0 0 0}]{hyperref}
\usepackage{wrapfig}
\usepackage{xcolor}
\usepackage[small]{caption}
\usepackage{multirow}

\usepackage{pbox}
\usepackage{hyperref}

\iftoggle{fullversion}{
  \usepackage[left=0.75in, right=0.75in, top=1.0in, bottom=1.0in]{geometry}
  \usepackage{titlefoot}
}{}

\usepackage[hang,flushmargin]{footmisc}
\usepackage{subfig}

\newcommand{\esm}[1]{\ensuremath{#1}}
\newcommand{\mr}[1]{\esm{\mathrm{#1}}}
\newcommand{\ms}[1]{\esm{\mathsf{#1}}}
\newcommand{\mb}[1]{\esm{\mathbf{#1}}}
\newcommand{\mc}[1]{\esm{\mathcal{#1}}}

\newcommand{\Z}{\esm{\mathbb{Z}}}
\newcommand{\F}{\esm{\mathbb{F}}}

\newcommand{\N}{\esm{\mathbb{N}}}

\newcommand{\abs}[1]{\esm{\left| #1 \right|}}
\newcommand{\set}[1]{\esm{\left\{ #1 \right\}}}

\newtheorem{theorem}{Theorem}[section]
\newtheorem{definition}[theorem]{Definition}

\newtheorem{claim}[theorem]{Claim}
\newtheorem{lemma}[theorem]{Lemma}

\DeclareMathOperator{\sign}{sign}

\newcommand{\zo}{\{0,1\}}

\newcommand{\ord}[1]{\esm{{#1}^{\mr{th}}}}
\newcommand{\hyb}{\ms{Hyb}}

\newcommand{\matrixind}[2]{#1^{(#2)}}
\newcommand{\getsr}{\xleftarrow{\textsc{r}}}
\newcommand{\src}{\mr{src}}
\newcommand{\dst}{\mr{dst}}

\newcommand{\dsrc}{\mc{D}_{\src}}
\newcommand{\ddst}{\mc{D}_{\dst}}

\newcommand{\dnw}{\textsc{nw}}
\newcommand{\dne}{\textsc{ne}}
\newcommand{\dn}{\textsc{n}}
\newcommand{\de}{\textsc{e}}
\newcommand{\dw}{\textsc{w}}
\newcommand{\ds}{\textsc{s}}
\newcommand{\dir}{\esm{\mathsf{dir}}}
\newcommand{\axistodir}{\esm{\mathsf{IndexToDirection}}}

\newcommand{\execr}{\ms{REAL}}
\newcommand{\execi}{\ms{IDEAL}}
\newcommand{\ppt}{\ms{PPT}}
\newcommand{\appc}{\stackrel{c}{\approx}}
\newcommand\numberthis{\addtocounter{equation}{1}\tag{\theequation}}
\newcommand{\view}{\ms{view}}
\newcommand{\viewi}[2]{\matrixind{\view}{#1}_{#2}}
\newcommand{\real}{\ms{real}}

\newcommand{\adv}{\ms{Adv}}
\newcommand{\advsec}{\matrixind{\adv}{\text{sec}}}
\newcommand{\advpriv}{\matrixind{\adv}{\text{priv}}}
\newcommand{\negl}{\ms{negl}}
\newcommand{\pir}{\ms{PIR}}
\newcommand{\hybrid}[1]{\ms{H}_{#1}}

\newcommand{\cunblind}{\esm{C^{\text{unblind}}}}
\newcommand{\tcunblind}{\esm{\tilde C^{\text{unblind}}}}
\newcommand{\bcunblind}{\esm{\bar C^{\text{unblind}}}}

\newcommand{\laffine}{\esm{L^{\text{affine}}}}
\newcommand{\hlaffine}{\esm{\hat L^{\text{affine}}}}
\newcommand{\blaffine}{\esm{\bar L^{\text{affine}}}}
\newcommand{\tlaffine}{\esm{\tilde L^{\text{affine}}}}

\newcommand{\lunblind}{\esm{L^{\text{unblind}}}}
\newcommand{\hlunblind}{\esm{\hat L^{\text{unblind}}}}
\newcommand{\blunblind}{\esm{\bar L^{\text{unblind}}}}
\newcommand{\tlunblind}{\esm{\tilde L^{\text{unblind}}}}

\newcommand{\prfkeylen}{\esm{\rho}}
\newcommand{\enckeylen}{\esm{\ell}}
\newcommand{\vorig}{\esm{n}}
\newcommand{\vcomp}{\esm{d}}
\newcommand{\diri}{\dir}

\newcommand{\para}[1]{\iftoggle{fullversion}{\paragraph{#1}}{\noindent \textbf{#1}}}

\def\yao{\ms{Yao}}
\def\piyao{\Pi_\yao}
\def\yaogarble{\yao.\ms{Garble}}
\def\yaoenc{\yao.\ms{Encode}}
\def\yaoeval{\yao.\ms{Eval}}

\def\affine{\ms{affine}}

\def\sk{\ms{sk}}
\def\poly{\ms{poly}}

\newcounter{protocounter}  %

\begin{document}

\iftoggle{fullversion}{
  \title{Privacy-Preserving Shortest Path Computation \\ {\normalsize (Extended Version)}}
\author{David J. Wu \qquad Joe Zimmerman \qquad J{\'e}r{\'e}my Planul
                    \qquad John C. Mitchell \\ \\
        Stanford University \\
        \texttt{\{dwu4, jzim, mitchell\}@cs.stanford.edu}, \texttt{jeremy.planul@ens-lyon.org}}

\unmarkedfntext{This is the extended version of a paper by the same name
that appeared at the {\it Network and Distributed System Security Symposium (NDSS)}
in February, 2016. Permission to freely reproduce all or part
 of this paper for noncommercial purposes is granted provided that
 copies bear this notice and the full citation on the first
 page. Reproduction for commercial purposes is strictly prohibited
 without the prior written consent of the Internet Society, the
 first-named author (for reproduction of an entire paper only), and
 the author's employer if the paper was prepared within the scope
 of employment.  \\
 NDSS '16, 21-24 February 2016, San Diego, CA, USA\\
 Copyright 2016 Internet Society, ISBN 1-891562-41-X\\
 http://dx.doi.org/10.14722/ndss.2016.23052}

}{
  \title{Privacy-Preserving Shortest Path Computation}
\author{\IEEEauthorblockN{David J. Wu, Joe Zimmerman,
        J{\'e}r{\'e}my Planul, John C. Mitchell}
        \IEEEauthorblockA{Stanford University \\
        \{dwu4, jzim, mitchell\}@cs.stanford.edu,
         jeremy.planul@ens-lyon.org}}

\IEEEoverridecommandlockouts
\makeatletter\def\@IEEEpubidpullup{9\baselineskip}\makeatother
\IEEEpubid{\parbox{\columnwidth}{Permission to freely reproduce all or part
    of this paper for noncommercial purposes is granted provided that
    copies bear this notice and the full citation on the first
    page. Reproduction for commercial purposes is strictly prohibited
    without the prior written consent of the Internet Society, the
    first-named author (for reproduction of an entire paper only), and
    the author's employer if the paper was prepared within the scope
    of employment.  \\
    NDSS '16, 21-24 February 2016, San Diego, CA, USA\\
    Copyright 2016 Internet Society, ISBN 1-891562-41-X\\
    http://dx.doi.org/10.14722/ndss.2016.23052
}
\hspace{\columnsep}\makebox[\columnwidth]{}}
\thispagestyle{empty}
}
\date{}
\maketitle

\begin{abstract}
Navigation is one of the most popular cloud computing services. But in
virtually all cloud-based navigation systems, the client must reveal her
location and destination to the cloud service provider in order to learn the
fastest route. In this work, we present a cryptographic protocol for
navigation on city streets that provides privacy for both the client's
location and the service provider's routing data. Our key ingredient is a
novel method for compressing the next-hop routing matrices in networks such as
city street maps. Applying our compression method to the map of Los Angeles,
for example, we achieve over tenfold reduction in the representation size. In
conjunction with other cryptographic techniques, this compressed
representation results in an efficient protocol suitable for fully-private
real-time navigation on city streets. We demonstrate the practicality of our
protocol by benchmarking it on real street map data for major cities such as
San Francisco and Washington,~D.C.
\end{abstract}

\section{Introduction}
\label{sec:introduction}

Location privacy is a major concern
among smartphone users, and there have been numerous controversies due to
companies tracking users' locations~\cite{AV11,Che11}. Among the
various applications that require location information, navigation is one of the most popular.
For example, companies such as Google, Apple, and Waze
have built traffic-aware navigation apps to provide users
with the most up-to-date routing information. But to use these services,
users must reveal their location and destination to the cloud service provider.
In doing so,
they may also reveal other sensitive information about their personal lives, such as their health
condition, their social and political affiliations, and more.

One way to provide location privacy is for the user to download the entire map
from the cloud service provider and then compute the best route locally
on her own mobile device.
Unfortunately, since service providers invest significant resources
to maintain up-to-date routing information, they are not incentivized to publish
their entire routing database in real-time.
Even in the case of a paid premium service, in which the service provider
does not derive compensation from learning the user's location data,
it is not obvious how to achieve fully-private navigation.
The user does not trust the cloud provider with her location data,
and the cloud provider does not trust the user with its up-to-date routing information,
so neither party has all of the data to perform the computation.
While general-purpose cryptographic tools such as multiparty computation
solve this problem in theory (see Section~\ref{sec:related-work}),
these protocols are prohibitively expensive in practice
for applications such as real-time navigation.

\para{Our results.} 
In this work, we present an efficient cryptographic protocol for
{\em fully-private} navigation:~the user
keeps private her location and destination,
and the service provider keeps private
its proprietary routing information (except for the routing information
associated with the specific path requested by the user and a few generic
parameters pertaining to the network).
We give a complete implementation of our protocol and
benchmark its performance on real street map data (Section~\ref{sec:experiments}).
Since our protocol is real-time (the user continues receiving directions
throughout the route), we benchmark the performance ``per hop'',
where each hop roughly corresponds to an intersection between streets.\footnote{In
a few cases, hops in our construction occur mid-street or in instances such as traffic circles.
These are rare enough that even in large cities such as Los Angeles,
the total number of hops along any route is less than 200.}
For cities such as San Francisco and Washington,~D.C.,
each hop in our protocol requires about 1.5 seconds
and less than $100$~KB of bandwidth.
In addition, before the protocol begins, we execute
a preprocessing step that requires bandwidth
in the tens of megabytes. Since this preprocessing step can be performed at any time,
in practice it would likely be run via a fast Wi-Fi connection,
before the mobile user needs the real-time navigation service,
and thus the additional cost is very modest.
To our knowledge,
ours is the first fully-private navigation
protocol efficient enough to be feasible in practice.

\para{Our technical contributions.}
In our work, we model street-map networks as graphs,
in which the nodes correspond
to street intersections, and edges correspond to streets.
In our model, we assume that the network topology is public (i.e.,
in the case of navigation on city streets,
the layout of the streets is publicly known).
However, only the service provider knows the up-to-date traffic conditions, and thus the
shortest path information. In this case, the server's ``routing information'' consists
of the weights (that is, travel times) on the edges in the network.

By modeling street-maps as graphs, we can easily construct
a straw-man private navigation protocol based on
symmetric private information retrieval (SPIR)~\cite{GIKM00,KO97,NP05}.
Given a graph $\mc{G}$ with $n$ nodes, the server first constructs a database with $n^2$ records, each
indexed by a source-destination pair $(s,t)$. The record indexed $(s,t)$ contains the shortest
path from $s$ to $t$. To learn the shortest path from $s$ to $t$, the client engages in SPIR
with the server for the record indexed $(s,t)$.
Security of SPIR implies that the client just learns the shortest path and the server learns
nothing. While this method satisfies the basic security requirements, its complexity
scales quadratically in the number of nodes in the graph. Due to the computational cost of
SPIR, this solution quickly becomes infeasible in the size of the graph.

Instead,
we propose a novel method to compress the routing information in street-map networks.
Specifically, given a graph $\mc{G}$ with $\vorig$ nodes, we define the next-hop routing
matrix $M \in \Z^{\vorig \times \vorig}$ for $\mc{G}$ to be the matrix
where each entry $M_{st}$
gives the index of the first node on the shortest path from node $s$ to node $t$.
To apply our compression method, we first preprocess the graph
(Section~\ref{sec:graph-compression}) such that each
entry in the next-hop routing matrix $M$ can be specified by two bits:
$M_{st} = (\matrixind{M}{\dne}_{st}, \, \matrixind{M}{\dnw}_{st})$ where
$\matrixind{M}{\dne}, \matrixind{M}{\dnw} \in \set{-1,1}^{\vorig \times \vorig}$.
We then compress $\matrixind{M}{\dne}$ by computing
a \emph{sign-preserving decomposition}: two matrices
$\matrixind{A}{\dne}, \matrixind{B}{\dne} \in \Z^{\vorig \times \vcomp}$ 
where $d \ll n$ such that
$\matrixind{M}{\dne} = \sign(\matrixind{A}{\dne} \cdot (\matrixind{B}{\dne})^T$).
We apply the same procedure to compress
the other component $\matrixind{M}{\dnw}$.
The resulting compression
is lossless, so there is no loss in accuracy in the
shortest paths after applying our transformation.
When applied to the road network for the city of Los Angeles, we obtain
over 10x reduction in the size of the representation.
Our compression method is highly parallelizable and by running our computation on
GPUs, we can compress next-hop matrices with close to 50 million elements
(for a 7000-node network) in under ten minutes.

Moreover, our compression method enables an efficient protocol for a fully-private
shortest path computation.
In our protocol, the rounds of interaction correspond to the nodes in the shortest path.
On each iteration of the protocol, the client learns the next
hop on the shortest path to its requested destination. Abstractly, if the client is currently
at a node $s$ and navigating to a destination $t$, then after one round of the 
protocol execution, the client should learn the next hop given by
$M_{st} = (\matrixind{M}{\dne}_{st}, \, \matrixind{M}{\dnw}_{st})$. Each round
of our protocol thus reduces to a two-party computation of the components
$\matrixind{M}{\dne}_{st}$ and $\matrixind{M}{\dnw}_{st}$.
Given our compressed representation of the next-hop routing matrices,
computing $\matrixind{M}{\dne}_{st}$ reduces to computing the sign of the inner product between
the $\ord{s}$ row of $\matrixind{A}{\dne}$ and the $\ord{t}$ row of $\matrixind{B}{\dne}$,
and similarly for $\matrixind{M}{\dnw}$.
In our construction, we give an efficient method for inner product evaluation based
on affine encodings, and use Yao's garbled circuits~\cite{Yao86,BHR12} to evaluate
the sign function. An important component of our protocol design is a novel way
of efficiently combining affine encodings and garbled circuits. Together, these methods
enable us to construct an efficient, fully-private navigation protocol.

\para{Other approaches.}
An alternative method for private navigation is to use
generic tools for two-party computation such as
Yao's garbled circuits~\cite{Yao86,BHR12} and Oblivious RAM (ORAM)~\cite{GO96,SDSFRYD13}.
While these approaches are versatile, they are often prohibitively expensive for city-scale
networks (in the case of Yao circuits), or do not provide strong security guarantees against
malicious clients (in the case of ORAM). For instance, the garbled-circuit approach
by Carter et al.~\cite{CMTB13,CLT14}
requires several minutes of computation to answer a single shortest path query in a road network with
just 100 nodes. Another generic approach combining garbled circuits and ORAM~\cite{LWNHS15}
requires communication on the order of GB and run-times ranging from tens of minutes to several hours
for a single query on a network with 1024 nodes. Thus, current
state-of-the-art tools for general two-party computation do not give a viable solution for
private navigation in city-scale networks. We survey other
related methods in Section~\ref{sec:related-work}.

\section{Preliminaries and Threat Model}
\label{sec:prelim}

We begin with some notation. 
For a positive integer $n$, let $[n]$ denote the
set of integers $\set{1,\ldots,n}$.
For two $\ell$-bit strings $x,y \in \zo^\ell$, we write $x \oplus y$ to denote
their bitwise XOR. For a prime $p$,
we write $\F_p$ to denote the finite field with $p$ elements,
and $\F_p^*$ to denote its multiplicative group.
Let $\mc{D}$ be a probability distribution. We write
$x \gets \mc{D}$ to denote that $x$ is drawn from $\mc{D}$.
Similarly, for a finite set $S$ we write $x \getsr S$
to denote that $x$ is drawn uniformly at random from $S$.
A function $f(\lambda)$ is negligible in a security parameter $\lambda$ if
$f = o(1/\lambda^c)$ for all $c \in \N$.

For two distribution ensembles
$\set{\mc{D}_1}_\lambda, \set{\mc{D}_2}_\lambda$, we write
$\set{\mc{D}_1}_\lambda \appc \set{\mc{D}_2}_\lambda$ to denote that
$\set{\mc{D}_1}_\lambda$ and $\set{\mc{D}_2}_\lambda$ are computationally
indistinguishable (i.e., no probabilistic polynomial-time algorithm
can distinguish them, except with probability negligible in
$\lambda$).
We write $\set{\mc{D}_1}_\lambda \equiv \set{\mc{D}_2}_\lambda$
to denote that $\set{\mc{D}_1}_\lambda$ and $\set{\mc{D}_2}_\lambda$
are identically distributed for all $\lambda$.
For a predicate $\mc{P}(x)$, we write $\mb{1}\{ \mc{P}(x) \}$ to denote the indicator
function for $\mc{P}(x)$, i.e., $\mb{1}\{ \mc{P}(x) \} = 1$ if and only
if $\mc{P}(x)$ is true, and otherwise, $\mb{1}\{ \mc{P}(x) \} = 0$. If
$\mc{G}$ is a directed graph, we write $(u,v)$ to denote the edge from
node $u$ to node $v$.

A function $F: \mc{K} \times \mc{X} \to \mc{Y}$ with key-space $\mc{K}$,
domain $\mc{X}$, and range $\mc{Y}$ is a PRF~\cite{GGM86} if no efficient adversary can
distinguish outputs of the PRF (with key $k \getsr \mc{K}$, evaluated on inputs
chosen adaptively by the adversary) from
the corresponding outputs of a truly random function from $\mc{X} \to \mc{Y}$.

\para{Threat model.}
We give a high-level survey of our desired security properties, and defer the details
to Section~\ref{sec:security-model}.
We operate in the two-party setting where both parties know the network topology as well
as a few generic parameters about the underlying graph structure (described concretely
in Section~\ref{sec:security-model}), but only
the server knows the weights (the routing information). The client
holds a source-destination pair.
At the end of the protocol execution, the client
learns the shortest path between its requested source and destination, while the server
learns nothing. The first property we require is privacy for the client's location.
Because of the sensitivity of location information, we require privacy to hold even against malicious
servers, that is, servers whose behavior can deviate from the protocol specification.

The second requirement is privacy for the server's routing information, which may contain
proprietary or confidential information. The strongest notion we can impose
is that at the end of the protocol execution,
the client does not learn anything more about the graph other than the shortest path between its
requested source and destination and some generic parameters associated with the underlying network.
While this property is not difficult to achieve
if the client is semi-honest (that is, the client adheres to the
protocol specification), in practice
there is little reason to assume that the client will behave this way.   %
Thus, we aim to achieve security against malicious clients. 
In our setting, we will show that a malicious client learns only the shortest path
from its requested source to its requested destination, except
for failure events that occur with probability at most $\approx 2^{-30}$.
For comparison, $2^{-30}$ is the probability that an adversary running in time $\approx 2^{50}$
is able to guess an 80-bit secret key.\footnote{
Even in the case of these low-probability failure events,
one can show that a malicious client only learns a bounded-length path
emanating from its requested source, though it may not be a shortest path
to any particular destination.}

To summarize, we desire a protocol that provides privacy against a malicious server and security 
against a malicious client. We note that our protocol does not protect against the
case of a server corrupting the map data; in practice, we assume that the map provider is trying to provide a useful
service, and thus is not incentivized to provide misleading or false navigation information.

\section{Graph Processing}
\label{sec:graph-compression}

As described in Section~\ref{sec:introduction}, we model street-map networks as
directed graphs,
where nodes correspond to intersections, and edges correspond to streets. To enable an
efficient protocol for fully-private shortest path computation, we first develop
an efficient method
for preprocessing and compressing the routing information in the network.
In this section, we first describe our preprocessing procedure, which consists of
two steps: introducing dummy nodes to constrain the out-degree of the graph,
and assigning a cardinal direction to each edge. Then,
we describe our method for compressing the routing information in the graph;
here, we exploit the geometric structure of the graph.

\begin{figure*}[t]
  \begin{center}
    \fbox{\includegraphics[height=5.2cm]{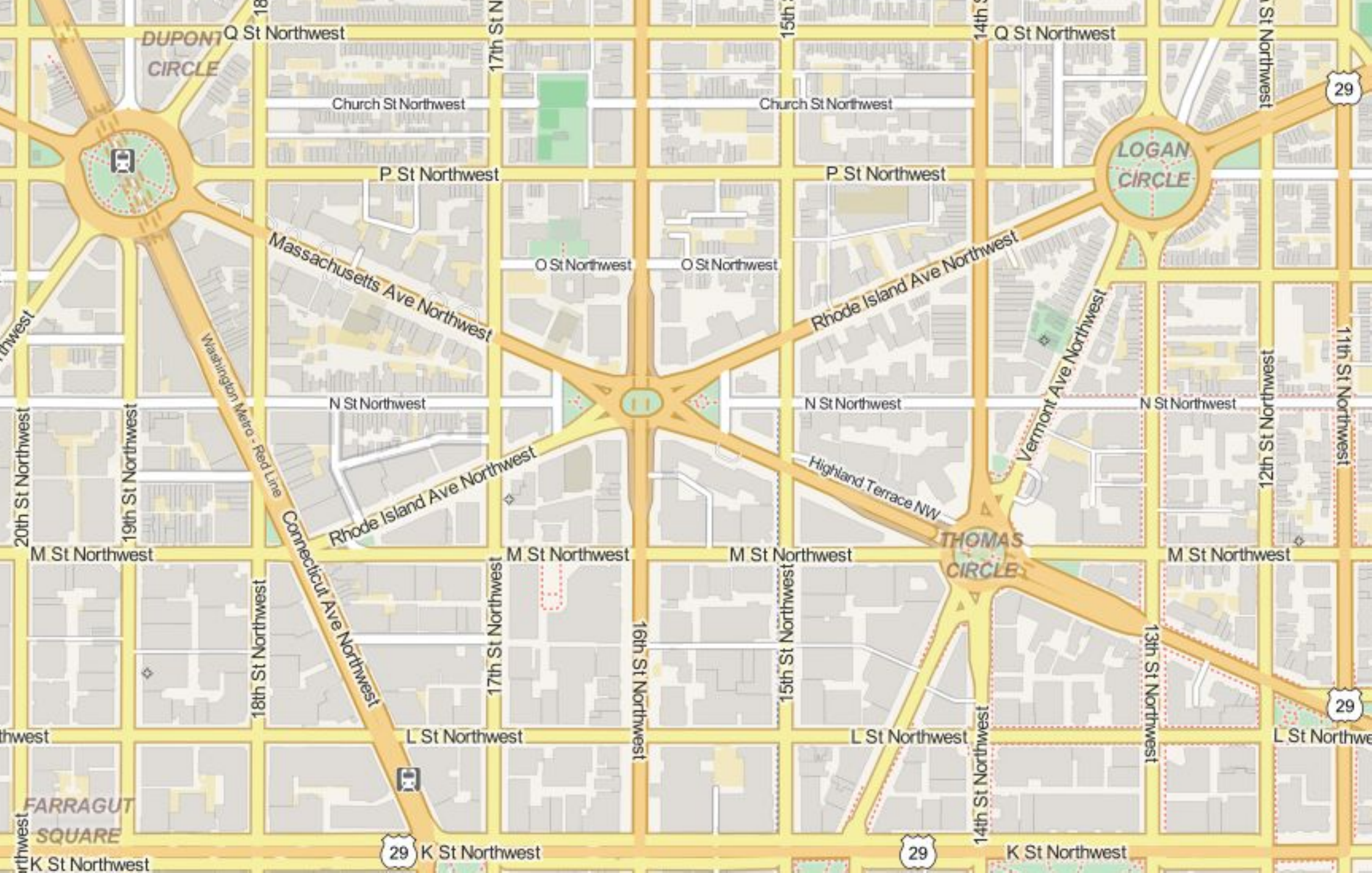}}
    \fbox{\includegraphics[height=5.2cm]{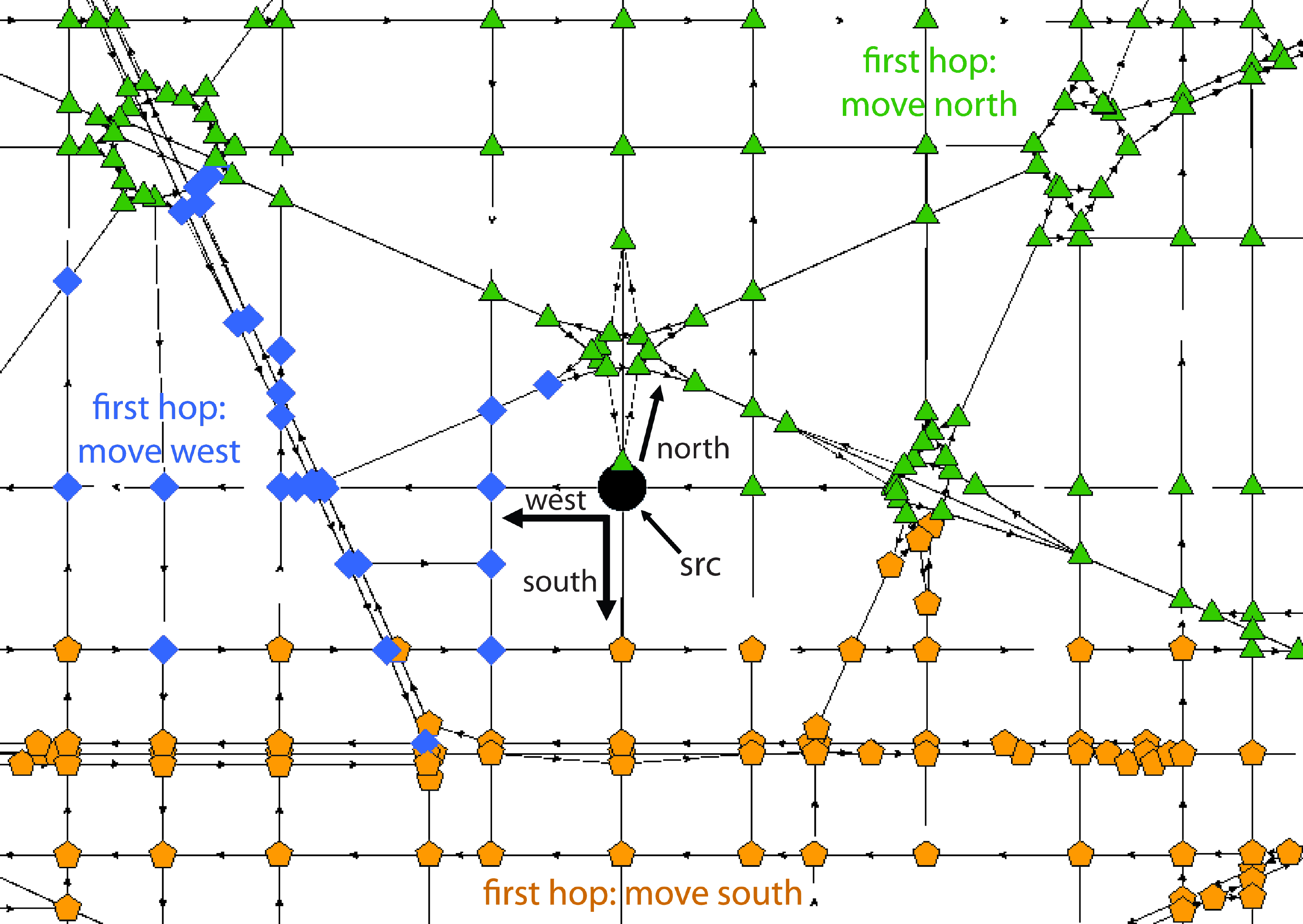}}
  \end{center}
  \caption{Subsection of map of Washington,~D.C. from OpenStreetMap~\cite{OpenStreetMap} (left)
           and visualization of the routing network after preprocessing (right). The visualization on
           the right shows the first hop of the shortest path from the source node (denoted by
           the circle) to all other nodes in the graph (denoted by a polygon). In this example,
           the source node has three neighbors: to the north, west, and south (as indicated in the diagram).
           If the first hop in the shortest path from the source to a node is to move north, then the node is
           represented by a green triangle. If the first hop is to move west, then the node is
           represented by a blue diamond, and if the first hop is to move south, then the node is
           represented by an orange pentagon.}
  \label{fig:map-preprocessing}
\end{figure*}

\para{Bounding the out-degree.} Let $\mc{G}$ be the directed graph representing
the road network. We assume that the nodes in $\mc{G}$ have
low out-degree. In a road network (see Figure~\ref{fig:map-preprocessing} for an example),
the nodes correspond to street intersections, and thus
typically have at most four outgoing edges, one in each cardinal direction. 
In the first step of our preprocessing procedure,
we take a weighted, directed graph $\mc{G}$ and transform it into a weighted, directed
graph $\mc{G}'$ where each node has maximum out-degree 4. Specifically, we start
by setting $\mc{G}' = \mc{G}$. Then, as long as there is a node $u \in \mc{G}'$
with neighbors $v_1,\ldots,v_\ell$ and $\ell>4$, we do the following. First, we add a new node $u'$ to $\mc{G}'$.
For $i \ge 4$, we add the edge $(u', v_i)$ to $\mc{G}'$ and remove the edge $(u,v_i)$ from
$\mc{G}'$. We also add a zero-weight edge from $u$ to $u'$ in $\mc{G}'$. By construction,
this transformation preserves the shortest-path between nodes in $\mc{G}$
and constrains the out-degree of all nodes in $\mc{G}'$ to 4.

\para{Orienting the edges.} In a road network,
we can associate each node by an $(x,y)$ pair in the coordinate plane
(for example, the node's latitude and longitude).
Consider a coordinate system aligned with
the cardinal directions: the $x$-axis corresponds to the east-west axis
and the $y$-axis corresponds to the north-south axis.
Then, for each node $u$ in the graph $\mc{G}$,
we associate each of its neighbors $v_i$ ($0 \le i < 4$)
with a direction $\dir_i \in \set{\textsc{n}, \textsc{e}, \textsc{s}, \textsc{w}}$
(for north, east, south, west, respectively)
relative to $u$. For a concrete example, refer
to the visualization of the preprocessed graph in Figure~\ref{fig:map-preprocessing}. Here,
the center node (labeled ``src'') has three neighbors, each of which is associated with
a cardinal direction: north, west, or south in this case. We define the orientation
of an edge to be the direction associated with the edge.

To determine the orientation of the edges in $\mc{G}$, we proceed as follows.
For each node $u \in \mc{G}$, we associate a \emph{unique} direction
$\dir_i \in \set{\textsc{n}, \textsc{e}, \textsc{s}, \textsc{w}}$
with each neighbor $v_i$ of $u$.
In assigning the four cardinal directions to each node's neighbors,
we would like to approximate the true geographical locations
of the nodes.
In our setting, we formulate this assignment as a
bipartite matching problem for each node $u$,
with $u$'s neighbors (at most 4) forming
one partition of the graph, and the four cardinal directions $\set{\textsc{n}, \textsc{e}, \textsc{s}, \textsc{w}}$
forming the other. 
We define the cost of a matching between a neighbor $v_i$ of $u$ and a direction $\dir_j$ to be
the angle formed by the vector from $u$ to $v_i$ and the unit vector aligned in the direction $\dir_j$.
In assigning directions to neighbors, we desire a matching
that minimizes the costs of the matched neighbors. Such
a matching can be computed efficiently using the Hungarian method \cite{KY55}.
In this way,
we associate a cardinal direction with each edge in $\mc{G}$.

\para{Compressing shortest paths.} Next, we describe a method for compressing
the next-hop routing matrix for a road network. Let $\mc{G}$ be a directed graph
with $n$ nodes and maximum out-degree 4. Using the method
described above, we associate a direction 
$\dir \in \set{\textsc{n}, \textsc{e}, \textsc{s}, \textsc{w}}$ with
each edge in $\mc{G}$.
Since there are four possible values for $\dir$, we can encode the direction
using exactly two bits $b_\dne$ and $b_\dnw$, where $b_\dne = 0$ if and only if
$\dir \in \set{\textsc{n}, \textsc{e}}$, and $b_\dnw = 0$ if and only if
$\dir \in \set{\textsc{n}, \textsc{w}}$. Intuitively, $b_\dne$ encodes the
direction with respect to the northwest-southeast axis while $b_\dnw$
encodes the direction with respect to the northeast-southwest axis.
Thus, for each node $u \in \mc{G}$, we associate a unique two-bit index
$(b_\dne, b_\dnw)$ with each of its outgoing edges.
For notational convenience, we define a function
$\axistodir$ that maps an index $(b_\dne, b_\dnw)$ to the corresponding
direction $\dir \in \set{\dn, \de, \ds, \dw}$.
\iftoggle{fullversion}{
  Specifically,
  \begin{equation}
  \begin{aligned}
    \axistodir(0, 0) &= \dn \\
    \axistodir(1, 0) &= \dw
  \end{aligned} \qquad
  \begin{aligned}
    \axistodir(0, 1) &= \de \\
    \axistodir(1, 1) &= \ds.
  \end{aligned}
  \label{eq:index-to-direction}
  \end{equation}
}{
  For example, $\axistodir(0, 0) = \dn$.
}

We next compute the shortest path $p_{st}$ between all source-destination pairs $(s,t)$ in $\mc{G}$.
In our implementation, we run Dijkstra's algorithm~\cite{Dij59} on each node in $\mc{G}$,
but the precise choice of shortest-path algorithm does not matter for our compression procedure,
as its cost is dominated by the other steps.
After computing all-pairs shortest paths in $\mc{G}$, we define the next-hop routing matrices
$\matrixind{M}\dne,\ \matrixind{M}\dnw \in \set{0,1}^{n \times n}$
for $\mc{G}$, where $(\matrixind{M}\dne_{st},\ \matrixind{M}\dnw_{st})$
encodes the direction of the first edge in the shortest path $p_{st}$. 

Just as the geometry of road networks enables us to
orient the edges, the geometry also suggests a method for compressing the next-hop
routing matrices. Take for example the road network
in Figure~\ref{fig:map-preprocessing}. From the visualization, we observe that
when the destination~$t$ lies to the north of the source~$s$, the first hop on the shortest
path is usually to take the edge directed north. In our framework, this means that
both $\matrixind{M}\dne_{st}$
and $\matrixind{M}\dnw_{st}$ are more likely to be~$0$ rather than~$1$.   %
Thus, by orienting the edges in the graph
consistently, we find that the resulting routing matrices
$\matrixind{M}\dne$ and $\matrixind{M}\dnw$
have potentially compressible structure.

To compress a matrix $M \in \zo^{\vorig \times \vorig}$, we first rescale the elements in 
$M$ to be in $\set{-1,1}$. Our goal is to find two matrices~$A,B \in \Z^{\vorig \times \vcomp}$
such that $\sign(A B^T) = M$ with $\vcomp < \vorig$.\footnote{This is not the same
as computing a low-rank approximation of $M$. Our goal
is to find low-rank matrices whose product
preserves the \emph{signs} of the entries of $M$.
In practice, the matrix $M$ is full-rank, and not well-approximated by a low-rank product.}
We can formulate the problem of computing $A$ and $B$ as an optimization problem
with objective function $J(A,B)$:
\begin{equation}
\label{eq:lbfgs-objective}
  J(A, B) = \sum_{j = 1}^\vorig \sum_{k = 1}^\vorig
    \ell \left( \left(A B^T \right)_{jk},\, M_{jk} \right),
\end{equation}
where $\ell(x,t)$ is a loss function.
A simple loss function
is the \mbox{0-1} loss function $\ell(x,t) = \mb{1}\{ \sign(x) \ne t \}$, which assigns
a uniform loss of 1 whenever the sign of the predicted value $x$ does not match the target value $t$.
However, from an optimization perspective, the \mbox{0-1} loss is not a good loss function
since it is non-convex and neither continuous nor differentiable.
Practitioners have instead used continuous convex approximations to the \mbox{0-1} loss,
such as the
SVM hinge loss $\ell_{\mr{hinge}}(x,t) = \max(0, 1 - tx)$~\cite{RVCPV04} 
and its quadratically smoothed variant, the modified Huber hinge loss~\cite{Zha04}:
\begin{equation}
\label{eq:huber-loss}
  \ell_{\mr{huber}}(x,t) = \begin{cases} \max(0, 1 - tx)^2 & t x \ge -1 \\ 
    -4 \cdot tx & \mr{otherwise}.
  \end{cases}
\end{equation}
In our setting, we use the modified Huber hinge loss $\ell_{\mr{huber}}$.
While $\ell_{\mr{huber}}$ is convex in the input $x$, it is
not convex in the optimization parameters $A,B$ (due to the matrix product), and so the
objective function $J(A,B)$ is not convex in $A,B$. Thus,
standard optimization algorithms like LBFGS~\cite{BLNZ95} are not guaranteed to find the
global optimum. The hope is that even a local optimum will correspond to
a low-rank, sign-preserving decomposition of the matrix $M$, and indeed, we
confirm this empirically.

When we perform the optimization using LBFGS, the matrices $A,B$
are real-valued. To obtain matrices over the integers, we scale the 
entries in $A,B$ by a constant factor and round.
The scaling factor is empirically chosen so as to preserve the relation
$\sign(AB^T) = M$. We describe this in greater detail in
Section~\ref{sec:experiments}.

\section{Private Navigation Protocol}
\label{sec:shortest-paths-protocol}

In this section, we describe our protocol for privately computing shortest
paths. First, we describe the cryptographic building blocks we employ
in our construction.

\para{Private information retrieval.} A computational private
information retrieval~(PIR)~\cite{CMS99,CGKS95,KO97,Cha04,GR05,Lip05,OS07} 
protocol is a two-party protocol between a sender who holds a database
$\mc{D} = \set{r_1, \ldots, r_n}$ and a receiver who holds an index
$i \in [n]$. At the conclusion of the PIR protocol, the receiver
learns $r_i$ while the sender learns nothing. A PIR protocol only ensures
privacy for the receiver's index (and not for the remaining records in
the sender's database).

\para{Oblivious transfer.} Similar to PIR, an 1-out-of-$n$ oblivious transfer~(OT)
protocol~\cite{NP99,NP01,NP05,Rab05} is a two-party protocol that allows the receiver to
privately retrieve a record $r_i$ from the sender who holds a database $\set{r_1,\ldots,r_n}$. 
In contrast with PIR, an OT protocol also provides privacy for the sender:
the receiver only learns its requested record $r_i$, and nothing else
about the other records.
Closely related is the notion of symmetric PIR (SPIR)~\cite{KO97,GIKM00,NP05}, which is functionally
equivalent to OT.

\para{Garbled circuits.}
Yao's garbled circuits~\cite{Yao86,LP09,BHR12} were initially developed for secure two-party computation.
The core of Yao's construction is an efficient transformation that takes a Boolean
circuit $C: \zo^n \to \zo^m$ and produces a garbled circuit $\tilde C$ along
with $n$ pairs of encodings $\set{k_i^0, k_i^1}_{i \in [n]}$. Then, for any input $x \in \zo^n$, the
combination of the garbled circuit $\tilde C$ and the encodings $S_x = \set{k_i^{x_i}}_{i \in [n]}$
(where $x_i$ denotes the $\ord{i}$ bit of $x$) enable one to compute $C(x)$,
and yet reveal nothing else about $x$.

\subsection{Protocol Design Overview}
We first give an intuitive overview of our fully-private navigation protocol.
As described in Section~\ref{sec:graph-compression}, we first preprocess
the network $\mc{G}$ to have maximum out-degree $d=4$ and then associate a cardinal direction
with each of the edges in $\mc{G}$. As in Section~\ref{sec:graph-compression},
let $(\matrixind{M}\dne,\matrixind{M}\dnw)$
be the precomputed next-hop routing matrices for $\mc{G}$, and let
$(\matrixind{A}{\dne},\matrixind{B}{\dne}), (\matrixind{A}{\dnw}, \matrixind{B}{\dnw})$
be the compressed representation 
of $\matrixind{M}{\dne}, \matrixind{M}{\dnw}$, respectively.

Our private shortest paths protocol is an iterative
protocol that reveals the shortest path from a source $s$ to a destination $t$
one hop at a time. When the client engages in the protocol with input
$(s,t)$, it learns which neighbor $v$ of $s$ is the next node on the shortest path from
$s$ to $t$. Then, on the next round of the protocol,
the client issues a query $(v,t)$ to learn the next node in the path, and
so on, until it arrives at the destination node $t$.
With this iterative approach, each round of our protocol
can be viewed as a two-party computation of the entry
$(\matrixind{M}\dne_{st}, \matrixind{M}\dnw_{st})$ from the next-hop
routing matrices. We give the full description of our private navigation
protocol in Figure~\ref{fig:protocol}, and sketch out
the important principles here. To simplify the presentation, we first
present the core building blocks that suffice for semi-honest security.
We then describe additional consistency checks that we introduce to obtain
security against a malicious client and privacy against a malicious
server.

\subsubsection{Semi-honest Secure Construction}
\label{sec:semi-honest-construction}

Abstractly, we can view each round of our protocol as computing the following two-party
functionality twice (once for $\matrixind{M}\dne$ and once for $\matrixind{M}\dnw$).
The server has two matrices $A,B \in \Z^{\vorig \times \vcomp}$, which we will
refer to as the source and destination matrices, respectively,
and the client has two indices $s,t \in [\vorig]$.
At the end of the protocol, the client should learn $\sign(\langle A_s, B_t \rangle)$,
where $A_s$ and $B_t$ are the $\ord{s}$ and $\ord{t}$ rows of $A$ and $B$, respectively.
The client should learn nothing else about $A$ and $B$, while
the server should not learn anything. Our protocol can thus be decomposed into
two components:
\begin{enumerate}
  \item Evaluation of the inner product $\langle A_s, B_t \rangle$ 
  between the \emph{source vector} $A_s$ and the
  \emph{destination vector} $B_t$.
  \item Determining the sign of $\langle A_s, B_t \rangle$.
\end{enumerate}
In the following, we will work over a finite field $\F_p$ large enough to contain the entries
in $A,B$. In particular, we view $A,B$ as $\vorig \times \vcomp$ matrices over $\F_p$.

\para{Evaluating the inner product.}
The first step in our protocol is evaluating the inner product between
the source vector $A_s$ and the destination vector $B_t$.
Directly revealing the value of $\langle A_s, B_t \rangle$ to the client, however,
leaks information about the entries in the compressed routing matrices 
$A,B$. To protect against this leakage, we instead reveal
a blinded version of the inner product.
Specifically, on each round of the protocol, the server chooses blinding factors
$\alpha \getsr \F_p^*$ and $\beta \getsr \F_p$. We then construct the protocol such that
at the end of the first step, the client learns the blinded value
$\alpha \langle A_s, B_t \rangle + \beta$ instead of
$\langle A_s, B_t \rangle$.

One candidate approach
for computing the blinded inner product is to use a garbled circuit. However,
while Yao's garbled circuits suffice for private evaluation of any two-party
functionality, when the underlying operations are more naturally expressed as 
addition and multiplication over $\F_p$, it is more convenient
to express the functionality in terms of an arithmetic circuit. In an arithmetic
circuit (over $\F_p$), the ``gates'' correspond to field operations (addition
and multiplication), and the values on the wires correspond to field elements.

In recent work, Applebaum et al.~\cite{AIK14} construct the analog of Yao's garbling
procedure for arithmetic circuits.
In particular, evaluating a function of the form
$f(x,y) = \langle x,y \rangle + \sum_{i \in [\vcomp]} z_i$,
where $x,y \in \F_p^\vcomp$ and each $z_i \in \F_p$ is a constant
can be done efficiently using the affinization gadgets
from~\cite[\S 5]{AIK14}. Specifically, for each $x_i,y_i$, we define
the following affine encoding functions $\laffine_{x_i}(x_i), \laffine_{y_i}(y_i)$:
\begin{align*}
  \laffine_{x_i}(x_i) &=
    \left( x_i - r_i^{(1)},\  x_i r_i^{(2)} + z_i + r_i^{(3)} \right) \\
  \laffine_{y_i}(y_i) &= 
    \left( y_i - r_i^{(2)},\  y_i r_i^{(1)} - r_i^{(1)} r_i^{(2)}- r_i^{(3)} \right), \numberthis \label{eq:arith-encoding}
\end{align*}
where $r_i^{(1)},r_i^{(2)},r_i^{(3)}$ are chosen uniformly from $\F_p$.
We will also write $\laffine_{x_i}(x_i; r_i), \laffine_{y_i}(y_i; r_i)$
to denote affine encodings of $x_i$ and $y_i$ using randomness $r_i \in \F_p^3$.
Given $\laffine_{x_i}(x_i)$ and $\laffine_{y_i}(y_i)$ for all $i \in [n]$, evaluating $f(x,y)$
corresponds to evaluating the expression
\begin{equation}
\label{eq:arithmetic-circuit-eval}
\sum_{i \in [n]} \left[ \laffine_{x_i}(x_i) \right]_1 \cdot \left[ \laffine_{y_i}(y_i) \right]_1 
   + \left[ \laffine_{x_i}(x_i) \right]_2  + \left[ \laffine_{y_i}(y_i) \right]_2,
\end{equation}
where we write $[\cdot]_i$ to denote the $\ord{i}$ component of a tuple. For notational
convenience, we also define $\laffine_x(x)$ and $\laffine_y(y)$ as
\begin{align*}
  \laffine_x(x) &=
    \left( \laffine_{x_1}(x_1), \ldots, \laffine_{x_\vcomp}(x_\vcomp) \right) \\
  \laffine_y(y) &= 
    \left( \laffine_{y_1}(y_1), \ldots, \laffine_{y_\vcomp}(y_\vcomp) \right). \numberthis
    \label{eq:arith-encoding-vector}
\end{align*}
Similarly, we write $\laffine_{x}(x; r), \laffine_{y}(y; r)$ to denote the affine
encoding of vectors $x, y \in \F_p^\vcomp$ using randomness $r \in \F_p^{3\vcomp}$.
The affine encodings $\laffine_x(x), \laffine_y(y)$ provides statistical
privacy for the input vectors $x,y$~\cite[Lemma 5.1]{AIK14}.

Next, we describe how these affine encodings can be used to compute the blinded
inner product in the first step of the protocol.
At the beginning of each round, the server chooses blinding factors
$\alpha \getsr \F_p^*$ and $\beta \getsr \F_p$.
Then, it constructs the affine encoding functions $\laffine_x, \laffine_y$
for the function $f_{\alpha, \beta}(x,y) = \langle \alpha x, y \rangle + \beta$
according to Eq.~\eqref{eq:arith-encoding}. Next, the server prepares
two encoding databases $\dsrc$ and $\ddst$ where the $\ord{s}$ record in $\dsrc$ 
consists of the affine encodings $\laffine_x(A_s)$ of each source vector, and the
$\ord{t}$ record in $\ddst$ consists of $\laffine_y(B_t)$ of
each destination vector.
To evaluate the blinded inner product, the client performs two
SPIR queries: one for the $\ord{s}$ record in $\dsrc$ to obtain the encodings of $A_s$ and
one for the $\ord{t}$ record in $\ddst$ to obtain the encodings of~$B_t$.\footnote{The
databases $\dsrc$ and $\ddst$ are each databases over $n$ records (as
opposed to $n^2$ in the straw-man protocol from Section~\ref{sec:introduction}).}
The client then evaluates the arithmetic circuit using
Eq.~\eqref{eq:arithmetic-circuit-eval} to obtain 
$z = f_{\alpha,\beta}(A_s, B_t)$.
To a malicious client, without knowledge of $\alpha$ or $\beta$,
the value
$f_{\alpha, \beta}(A_s, B_t)$ appears uniform over $\F_p$ and
independent of $A_s,B_t$.

\para{Determining the sign.}
To complete the description, it remains to describe a way for the client to learn
the sign of the inner product $\langle A_s, B_t \rangle$. The client has
the value $z = \alpha \langle A_s , B_t \rangle + \beta$ from the output of the arithmetic
circuit while the server knows the blinding factors $\alpha, \beta$. Since computing
the sign function is equivalent to performing a comparison, arithmetic circuits are
unsuitable for the task. Instead, we construct a separate Yao circuit to unblind the inner
product and compare it against zero. More specifically, let
$g(x, \gamma, \delta) = \mb{1}\{ [\gamma x + \delta]_p > 0\}$, where $[\cdot]_p$
denotes reduction modulo $p$, with output in the interval
$(-p/2, p/2)$. Then,
\[ g(z, \alpha^{-1}, -\alpha^{-1}\beta) = \sign(A_s, B_t). \]
To conclude the protocol, the server garbles a Boolean circuit $\cunblind$ for
the unblinding function $g$ to obtain a garbled circuit $\tcunblind$ along with a set
of encodings $\lunblind$. It sends the garbled circuit to the client, along with
encodings of the unblinding coefficients
$\gamma = \alpha^{-1}, \delta = \alpha^{-1} \beta$ to the
client. The client engages in 1-out-of-2 OTs to obtain the input
encodings of $z$, and evaluates the garbled circuit
$\tcunblind$ to learn $\sign(\langle A_s, B_t \rangle)$.

\subsubsection{Enforcing Consistency for Stronger Security}
\label{sec:ensuring-consistency}

As described, the protocol
reveals just a single edge in the shortest path. Repeated iteration of the protocol
allows the client to learn the full shortest path. Moreover, since the server's view of the protocol execution
consists only of its view in the PIR and OT protocols, privacy of these underlying 
primitives ensures privacy of the client's location, even against a malicious
server.\footnote{While a malicious
server can send the client malformed circuits or
induce selective failure attacks, the server does not receive any output during
the protocol execution nor does the client abort the protocol when malformed input
is received. Thus, we achieve privacy against a malicious server.}

Security for the server only holds if the client
follows the protocol and makes consistent queries on each round.
However, a malicious client can request the shortest path
for a different source and/or destination
on each round, thereby allowing it to learn edges along arbitrary
shortest paths of its choosing. To protect against a malicious client, we
bind the client to making \emph{consistent} queries across consecutive
rounds of the protocol. We say that a sequence of source-destination queries
$(s_1,t_1),\ldots,(s_\ell,t_\ell)$ is \emph{consistent}
if for all $i \in [\ell]$, $t_1 = t_i$, and $s_{i+1} = v_i$ where $v_i$ is the first
node on the shortest path from $s_i$ to $t_i$.

\para{Consistency for the destinations.}
To bind the client to a single destination, we do the following. At the beginning
of the protocol, for each
row $i \in [\vorig]$ in $\ddst$, the server chooses a symmetric encryption key $k_{\dst, i}$.
Then, on each round of the protocol, 
it encrypts the $\ord{i}$ record in $\ddst$ with the key $k_{\dst, i}$.
Next, at the beginning of the protocol, the client OTs for the key $k_{\dst, t}$
corresponding to its destination $t$. Since this step is performed only once at
the beginning of the protocol, the only record in $\ddst$ that the client can decrypt is
the one corresponding to its original destination. Because each record in $\ddst$ is
encrypted under a different key, the client can use a PIR protocol
instead of an SPIR protocol when requesting the record from $\ddst$.

\para{Consistency for the sources.}
Maintaining consistency between the source queries is more challenging because the source
changes each round. We use the fact that the preprocessed
graph has out-degree at most four. Thus,
on each round, there are at most
four possible sources that can appear in a consistent query in the next round.

Our construction uses a semantically-secure symmetric encryption scheme
$(\ms{Enc}, \ms{Dec})$ with
key-space $\zo^\enckeylen$, and a PRF $F$ with domain $\set{\dn, \de, \ds, \dw}$
and range $\zo^\enckeylen$. On each round of the protocol, the server
generates a new set of source keys
$k_{\src, 1}, \ldots, k_{\src, n} \in \zo^\enckeylen$ for
encrypting the records in $\dsrc$ in the \emph{next} round of the protocol.
The server also chooses four PRF keys $k_\dne^0, k_\dne^1, k_\dnw^0, k_\dnw^1$,
which are used to derive directional keys $k_\dn, k_\de, k_\ds, k_\dw$.
Next, for each node $v \in [n]$ in $\dsrc$, let $v_\dir$ be the neighbor
of $v$ in direction $\dir \in \set{\dn, \de, \ds, \dw}$ (if there is one).
The server augments the $\ord{v}$ record in $\dsrc$ with an encryption of
the source key $k_{\src, v_{\dir}}$ under the directional key $k_\dir$.

When the client requests record $v$
from $\dsrc$, it also obtains encryptions of the keys of the neighbors of $v$
for the \emph{next} round of the protocol. By ensuring the client
only learns one of the directional keys, it will only be able to
learn the encryption key for a single source node on the next round of the protocol.
We achieve this by including the PRF keys $k_\dne^0, k_\dne^1, k_\dnw^0, k_\dnw^1$
used to derive the directional keys as input to the garbled circuit.
Then, in addition to outputting the direction, the garbled circuit also outputs the
subset of PRF keys needed to derive exactly one of the
directional keys $k_\dn, k_\de, k_\ds, k_\dw$.
This ensures that the client has at most one source key in the next round of the
protocol.

\para{Consistency within a round.} In addition to ensuring consistency between
consecutive rounds of the protocol, we also require that the client's input to
the garbled circuit is consistent with the output it obtained from evaluating 
the affine encodings. To enforce this, we use the fact that the entries of the
routing matrices $A,B$ are bounded: there exists $\tau$ such that
$\langle A_s, B_t \rangle \in [-2^\tau, 2^\tau]$ for all~$s,t \in V$.
Then, in our construction, we choose the size of the finite field $\F_p$ to be much larger
than the size of the interval $[-2^{\tau}, 2^\tau]$. Recall that the arithmetic circuit computes
a blinded inner product $z \gets \alpha \langle A_s, B_t \rangle + \beta$ where $\alpha, \beta$ are
uniform in $\F_p^*$ and $\F_p$, respectively. To unblind the inner product, the server constructs
a garbled circuit that first evaluates the function $g_{\gamma,\delta}(z) = \gamma z + \delta$
with $\gamma = \alpha^{-1}$ and $\delta = -\alpha^{-1}\beta$. By construction,
$\gamma$ is uniform over $\F_p^*$ and $\delta$ is uniform over $\F_p$. Thus, using
the fact that $\set{g_{\gamma,\delta}(z) \mid \gamma \in \F_p^*, \delta \in \F_p}$ is a pairwise independent
family of functions, we conclude that
the probability that $g_{\gamma,\delta}(z') \in [-2^\tau, 2^\tau]$ is precisely
$2^{\tau+1} / p$ for all $z' \in \F_p$. By choosing $p \gg 2^{\tau + 1}$, we can ensure that
the adversary cannot successfully cheat except with very small probability.

Lastly, we remark that when the client issues a query $(s,t)$ where $s = t$, the protocol should
not reveal the key for any other node in the graph. To address this, we also introduce an equality
test into the garbled circuit such that on input $s = t$, the output is $\perp$. 
We give a complete specification of the neighbor-computation function that incorporates these
additional consistency checks in Figure~\ref{fig:neighbor-computation}.

\begin{figure} \small
\begin{framed}
\noindent \textbf{Inputs:} Tuples 
$(z_\dne, \gamma_\dne, \delta_\dne), (z_\dnw, \gamma_\dnw, \delta_\dnw) \in \F_p^3$,
PRF keys $k_\dne^0, k_\dne^1, k_\dnw^0, k_\dnw^1 \in \zo^{\prfkeylen}$,
and the source and destination nodes $s,t \in [\vorig]$.
The bit-length $\tau$ is public and fixed (hard-wired into $g$).
\\ \\
\noindent \textbf{Operation of $g$:}
\begin{itemize}
\item If $s = t$, then output $\perp$.
\item If $[\gamma_\dne z_\dne + \delta_\dne]_p \notin [-2^\tau, 2^\tau]$
or $[\gamma_\dnw z_\dnw + \delta_\dnw]_p \notin [-2^\tau, 2^\tau]$, output $\perp$.
\item Let $b_\dne = \mb{1}\{ [\gamma_\dne z_\dne + \delta_\dne]_p > 0 \}$,
and let $b_\dnw = \mb{1}\{ [\gamma_\dnw z_\dnw + \delta_\dnw]_p > 0 \}$.
Output $(b_\dne, b_\dnw, k_\dne^{b_\dne}, k_\dnw^{b_\dnw})$.
\end{itemize}
\end{framed}
\caption{Neighbor-computation function $g$ for the private routing protocol.}
\label{fig:neighbor-computation}
\end{figure}

\begin{figure*}
  \iftoggle{fullversion}{
    \begin{minipage}{\linewidth}
    \setcounter{mpfootnote}{\value{footnote}}
    \renewcommand{\thempfootnote}{\arabic{mpfootnote}}
  }{}
  \begin{framed}
  \small
  Fix a security parameter $\lambda$ and a statistical security parameter $\mu$.
  Let $\mc{G} = (V,E)$ be a weighted directed graph with $n$ vertices,
  such that the out-degree of every vertex is at most~4.
  The client's input to the protocol consists of two nodes, $s, t \in V$,
  representing the source and destination of the shortest path the client is requesting.
  The server's inputs are the compressed routing matrices
  $\matrixind{A}\dne, \matrixind{B}\dne, \matrixind{A}\dnw, \matrixind{B}\dnw \in \Z^{\vorig \times \vcomp}$
  (as defined in Section \ref{sec:graph-compression}).
  \\\\
  We assume the following quantities are public and known to both the client and the
  \iftoggle{fullversion}{server:
    \begin{itemize}[noitemsep]
      \item The structure of the graph $\mc{G}$, but not the edge weights.
      \item The number of columns $\vcomp$ in the compressed routing matrices.
      \item A bound on the bit-length $\tau$ of the values in the products
      $\matrixind{A}\dne \cdot (\matrixind{B}\dne)^T$ and
      $\matrixind{A}\dnw \cdot (\matrixind{B}\dnw)^T$.
      \item The total number of rounds $R$.
    \end{itemize}
  }{server:~the structure of the graph $\mc{G}$ (but not the edge weights); the
     number of columns $\vcomp$ in the compressed routing matrices; a 
     bound on the bit-length $\tau$ of the values in the products
     $\matrixind{A}\dne \cdot (\matrixind{B}\dne)^T$ and
     $\matrixind{A}\dnw \cdot (\matrixind{B}\dnw)^T$;  and the total number of rounds $R$.
     \\
  }

  In the following description,
  let $(\ms{Enc}, \ms{Dec})$ be a CPA-secure symmetric
  encryption scheme with key space $\zo^\enckeylen$, and
  let $F: \zo^\prfkeylen \times \set{\dn, \de, \ds, \dw} \to \zo^{\enckeylen}$
  be a PRF (where $\enckeylen,\prfkeylen = \poly(\lambda)$).
  Fix a prime-order finite field $\F_p$ such that $p > 2^{\tau + \mu + 1}$.
  \\\\
  \textbf{Setup:}
  \begin{enumerate}[noitemsep]
  \item For each $i \in [n]$, the server chooses independent symmetric encryption keys
  $\matrixind{k}{1}_{\src,i},\ k_{\dst,i} \getsr \zo^{\enckeylen}$.

  \item The client and the server engage in two 1-out-of-$n$ OT protocols with the client
  playing the role of the receiver:
  \begin{itemize}[noitemsep]
  \item The client requests the $\ord{s}$ record from the server's database
  $(\matrixind{k}{1}_{\src,1}, \ldots, \matrixind{k}{1}_{\src,n})$, receiving
  a value   %
  $\matrixind{\hat k}{1}_{\src}$.
  \item The client requests the $\ord{t}$ record from the server's database $(k_{\dst,1}, \ldots, k_{\dst,n})$,
  receiving a value $\hat k_{\dst}$.
  \end{itemize}
  \end{enumerate}

  \textbf{For each round $r = 1, \ldots, R$ of the protocol:}

  \begin{enumerate}[noitemsep]
  \item The server chooses blinding factors $\alpha_\dne, \alpha_\dnw \getsr \F_p^*$
  and $\beta_\dne, \beta_\dnw \getsr \F_p$. Next, let $\gamma_\dne = \alpha_\dne^{-1}$
  and $\delta_\dne = -\alpha_\dne^{-1} \beta_\dne \in \F_p$. Define $\gamma_\dnw$ and
  $\delta_\dnw$ analogously.

  \item Let $f_\dne, f_\dnw: \F_p^\vcomp \times \F_p^\vcomp \to \F_p$ where
  $f_\dne(x, y) = \langle \alpha_\dne x, y \rangle + \beta_\dne$ and
  $f_\dnw(x, y) = \langle \alpha_\dnw x, y \rangle + \beta_\dnw$.
  The server then does the following:
  \begin{itemize}
  \item Apply the affine encoding algorithm
  (Eq.~\ref{eq:arith-encoding}) to $f_\dne$ to obtain encoding functions
  $\laffine_{\dne, x}, \laffine_{\dnw, y}$, for the inputs $x$ and $y$, respectively.
  \item Apply the affine encoding algorithm to $f_\dnw$ to obtain
  encoding functions $\laffine_{\dnw, x}, \laffine_{\dnw, y}$.
  \end{itemize}

  \item Let $\cunblind$ be a Boolean circuit for computing the neighbor-computation
  function in Figure~\ref{fig:neighbor-computation}. The server runs Yao's garbling
  algorithm on $\cunblind$ to obtain a garbled circuit $\tcunblind$ along with
  encoding functions $\lunblind_{x}$, for each of the inputs $x$ to the neighbor-computation
  function in Figure~\ref{fig:neighbor-computation}.

  \item
    The server chooses symmetric encryption keys
    $\matrixind{k}{r+1}_{\src,1},\ldots,\matrixind{k}{r+1}_{\src,\vorig} \getsr \zo^{\enckeylen}$.
    These are used to encrypt the
  contents of the source database on the next round of the protocol.

  \item \label{enum:neighbor-key-derivation} The server chooses four PRF
  keys $k_\dne^0, k_\dne^1, k_\dnw^0, k_\dnw^1 \getsr \zo^{\prfkeylen}$,
  two for each axis. Then, the server defines the encryption keys for each direction
  as follows:\iftoggle{fullversion}{{\scriptsize \footnote{An alternative approach
  is for the server to choose the encryption keys $k_\dn, k_\de, k_\ds, k_\dw$
  uniformly at random from $\zo^\enckeylen$ instead of using the key-derivation
  procedure. In this case, the neighbor-computation function
  (Figure~\ref{fig:neighbor-computation}) would be modified
  to take as input the encryption keys $k_\dn, k_\de, k_\ds, k_\dw$ rather than the PRF keys
  $k_\dne^0, k_\dne^1, k_\dnw^0, k_\dnw^1$, and would output a single encryption
  key. While this approach is conceptually simpler, the resulting neighbor-computation
  circuits are slightly larger. In our implementation, we use the key-derivation procedure
  shown in the figure.
  }}}{}
  \iftoggle{fullversion}{
    \[
      \begin{aligned}
        k_\dn &= F(k_\dne^0, \dn) \oplus F(k_\dnw^0, \dn) \\
        k_\ds &= F(k_\dne^1, \ds) \oplus F(k_\dnw^1, \ds)
      \end{aligned} \qquad \qquad
      \begin{aligned}
        k_\de &= F(k_\dne^0, \de) \oplus F(k_\dnw^1, \de) \\
        k_\dw &= F(k_\dne^1, \dw) \oplus F(k_\dnw^0, \dw).
      \end{aligned}
    \]
  }{
    \[
        k_\dn = F(k_\dne^0, \dn) \oplus F(k_\dnw^0, \dn), \quad
        k_\de = F(k_\dne^0, \de) \oplus F(k_\dnw^1, \de), \quad
        k_\ds = F(k_\dne^1, \ds) \oplus F(k_\dnw^1, \ds), \quad
        k_\dw = F(k_\dne^1, \dw) \oplus F(k_\dnw^0, \dw).
    \]
  }

\iftoggle{fullversion}{
  \setcounter{protocounter}{\theenumi}

  \end{enumerate}
  \end{framed}

  \caption{The fully-private routing protocol,
  as outlined in Section~\ref{sec:shortest-paths-protocol}. The protocol description continues
  on the next page.}
  \label{fig:protocol}
  \ContinuedFloat

  \setcounter{footnote}{\value{mpfootnote}}
  \end{minipage}
  \end{figure*}

  \renewcommand{\thefigure}{\arabic{figure} (Continued)}
  \begin{figure*}
  \small
  \begin{framed}
  \begin{enumerate}
  \setcounter{enumi}{\theprotocounter}
}{}

  \item The server prepares the source database $\dsrc$ as follows. For each node $u \in [n]$, 
  the $\ord{u}$ record in $\dsrc$ is an encryption under $\matrixind{k}{r}_{\src, u}$ of the following:
  \begin{itemize}[noitemsep]
    \item The arithmetic circuit encodings $\laffine_{\dne, x}(\matrixind{A}\dne_u$),
    $\laffine_{\dnw,x}(\matrixind{A}\dnw_u)$ of the source vectors
    $\matrixind{A}\dne_u$ and $\matrixind{A}\dnw_u$.

    \item The garbled circuit encodings $\lunblind_s(u)$ of the source node $u$.

    \item Encryptions of the source keys for the neighbors of $u$ in the next round
    of the protocol under the direction keys:
    \[ \kappa_\dn = \ms{Enc}(k_\dn, \matrixind{k}{r+1}_{\src, v_\dn}), \quad
       \kappa_\de = \ms{Enc}(k_\de, \matrixind{k}{r+1}_{\src, v_\de}), \quad
       \kappa_\ds = \ms{Enc}(k_\ds, \matrixind{k}{r+1}_{\src, v_\ds}), \quad
       \kappa_\dw = \ms{Enc}(k_\dw, \matrixind{k}{r+1}_{\src, v_\dw}), \]
    where $v_\dn, v_\de, v_\ds, v_\dw$ is the neighbor of $u$ to the north, east,
    south, or west, respectively. If $u$ does not have a neighbor in a
    given direction $\diri \in \set{\dn, \de, \ds, \dw}$, then define
    $\matrixind{k}{r+1}_{\src, v_\diri}$ to be the all-zeroes string $0^\enckeylen$.
  \end{itemize}
  
  \item The server prepares the destination database $\ddst$ as follows.
  For each node $u \in [n]$, the $\ord{u}$ record in $\ddst$ is an encryption
  under $k_{\dst, u}$ of the following:
  \begin{itemize}[noitemsep]
    \item The arithmetic circuit encodings $\laffine_{\dne, y}(\matrixind{B}\dne_u$),
    $\laffine_{\dnw,y}(\matrixind{B}\dnw_u)$ of the destination vectors
    $\matrixind{B}\dne_u$ and $\matrixind{B}\dnw_u$.

    \item The garbled circuit encodings $\lunblind_t(u)$ of the destination node $u$.
  \end{itemize}

  \item The client and server engage in two PIR protocols with the client playing role of receiver:
  \begin{itemize}[noitemsep]
    \item The client requests record $s$ from the server's database
    $\dsrc$ and obtains a record $\hat c_{\src}$.
    \item The client requests record $t$ from the server's database
    $\ddst$ and obtains a record $\hat c_{\dst}$.
  \end{itemize}

\item The client decrypts the records:
$\hat r_{\src} \gets \ms{Dec}(\matrixind{\hat k}{r}_{\src}, \hat c_{\src})$ and
$\hat r_{\dst} \gets \ms{Dec}(\hat k_{\dst}, \hat c_{\dst})$:
\begin{itemize}[noitemsep]
\item  It parses $\hat r_{\src}$ into two sets of
arithmetic circuit encodings $\hlaffine_{\dne, x}$ and
$\hlaffine_{\dnw, x}$, a set of garbled circuit encodings $\hlunblind_s$,
and four encryptions
$\hat \kappa_\dn, \hat \kappa_\de, \hat \kappa_\ds, \hat \kappa_\dw$
of source keys for the next round.

\item It parses $\hat r_{\dst}$ into two sets of arithmetic circuit
encodings for $\hlaffine_{\dne, y}$ and $\hlaffine_{\dnw, y}$,
and a set of garbled circuit encodings $\hlunblind_t$.
\end{itemize}

Using the encodings $\hlaffine_{\dne, x}$ and $\hlaffine_{\dne, y}$,
the client evaluates the arithmetic circuit (Eq.~\ref{eq:arithmetic-circuit-eval})
to learn $\hat z_\dne$. Similarly, using the encodings
$\hlaffine_{\dnw,x}$ and $\hlaffine_{\dnw,y}$, the server evaluates to learn
$\hat z_\dnw$. If the parsing of $\hat r_{\src}$
or $\hat r_{\dst}$ fails or the arithmetic circuit encodings are malformed, the client
sets $\hat z_\dne, \hat z_\dnw \getsr \F_p$.

\iftoggle{fullversion}{}{
    \setcounter{protocounter}{\theenumi}

  \end{enumerate}
  \end{framed}

  \caption{The fully-private navigation protocol,
  as outlined in Section~\ref{sec:shortest-paths-protocol}. The protocol description continues
  on the next page.}
  \label{fig:protocol}
  \ContinuedFloat
  \end{figure*}

  \renewcommand{\thefigure}{\arabic{figure} (Continued)}
  \begin{figure*}
  \small
  \begin{framed}
  \begin{enumerate}
  \setcounter{enumi}{\theprotocounter}
}

\item The client engages in a series of 1-out-of-2 OTs with the server to
obtain the garbled circuit encodings $\lunblind_{z_\dne}(\hat z_\dne)$
and $\lunblind_{z_\dnw}(\hat z_\dnw)$
of $\hat z_\dne$ and $\hat z_\dnw$, respectively. Let
$\hlunblind_{z_\dne}$ and $\hlunblind_{z_\dnw}$ denote the encodings
the client receives.

\item The server sends to the client the garbled circuit $\tcunblind$ and
encodings of the unblinding coefficients
\[ \lunblind_{\gamma_\dne}(\gamma_\dne),\ 
   \lunblind_{\gamma_\dnw}(\gamma_\dnw),\ 
   \lunblind_{\delta_\dne}(\delta_\dne),\ 
   \lunblind_{\delta_\dnw}(\delta_\dnw),
\]
as well as encodings of the PRF keys
\[ 
  \lunblind_{k_\dne^0}(k_\dne^0), \
  \lunblind_{k_\dne^1}(k_\dne^1), \
  \lunblind_{k_\dnw^0}(k_\dnw^0), \
  \lunblind_{k_\dnw^1}(k_\dnw^1).
\]

\iftoggle{fullversion}{
  \setcounter{protocounter}{\theenumi}

  \end{enumerate}
  \end{framed}

  \caption{The fully-private routing protocol,
    as outlined in Section~\ref{sec:shortest-paths-protocol}.
    The protocol description continues on the next page.}
  \ContinuedFloat
  \end{figure*}

  \begin{figure*}
  \small
  \begin{framed}
  \begin{enumerate}
  \setcounter{enumi}{\theprotocounter}
}{}

\item The client evaluates the garbled circuit $\tcunblind$.
If the garbled circuit evaluation is successful and the client
obtain outputs $(\hat{b}_\dne, \hat{b}_\dnw, \hat{k}_\dne, \hat{k}_\dnw)$,
then the client computes a direction
$\dir = \axistodir(\hat b_\dne, \hat b_\dnw) \in \set{\dn, \de, \ds, \dw}$
(\iftoggle{fullversion}{Eq.~\eqref{eq:index-to-direction}}
{Section~\ref{sec:graph-compression}}).
\begin{enumerate}[noitemsep]
  \item The client computes the direction key
  $\hat k_\dir = F(\hat k_\dne, \dir) \oplus F(\hat k_\dnw, \dir)$.
  Next, the client decrypts the encrypted source key $\hat \kappa_\dir$
  to obtain the source key 
  $\matrixind{\hat k}{r+1}_\src = \ms{Dec}(\hat k_\dir, \hat \kappa_\dir)$
  for the next round of the protocol.

  \item Let $v_\dir$ be the neighbor of $s$
  in the direction given by $\dir$ (define $v_\dir$ to be $\perp$ if $s$ does
  not have a neighbor in the direction $\dir$). If $v_\dir \ne\ \perp$, the
  client outputs $v_\dir$ and updates $s = v_\dir$. Otherwise, if 
  $v_\dir =\ \perp$, the client outputs $\perp$ and leaves $s$ unchanged.
\end{enumerate}
If the OT for the input wires to the garbled circuit fails,
the garbled circuit evaluation fails, or the output of the garbled
circuit is $\perp$, then the client outputs $\perp$, but continues with the 
protocol: it leaves $s$ unchanged and sets $\matrixind{\hat k}{r+1}_{\src} \getsr \zo^\enckeylen$.

\end{enumerate}
\end{framed}

\caption{The fully-private routing protocol,
  as outlined in Section~\ref{sec:shortest-paths-protocol}.}
\end{figure*}

\renewcommand{\thefigure}{\arabic{figure}}

\subsection{Security Model}
\label{sec:security-model}

In this section, we formally specify our security model.
To define and argue the security of our protocol, we compare the protocol execution in the real-world (where
the parties interact according to the specification given in Figure~\ref{fig:protocol}) to
an execution in an ideal world where the parties have access to a trusted party that computes
the shortest path. Following the conventions
in~\cite{Can06}, we view the protocol execution as occurring in the
presence of an adversary $\mc{A}$ and coordinated by an environment
$\mc{E} = \set{\mc{E}}_\lambda$ (modeled as a family of polynomial
size circuits parameterized by a security parameter $\lambda$). The
environment $\mc{E}$ is responsible for choosing the inputs to the protocol
execution and plays the role of distinguisher between the real and ideal experiments.

As specified in Figure~\ref{fig:protocol}, we assume
that the following quantities are public to the protocol execution:
the topology of the network $\mc{G} = (V,E)$, the number of columns $\vcomp$
in the compressed routing matrices,
a bound on the bit-length $\tau$ of the values in the matrix products
$A^{(\dne)} \cdot (B^{(\dne)})^T$ and $A^{(\dnw)} \cdot (B^{(\dnw)})^T$, and
the total number of rounds $R$ (i.e., the number of hops in the longest possible
shortest path). We now define the real and ideal models of execution.

\begin{definition}[Real Model of Execution]
\label{def:real-model}
\normalfont
Let $\pi$ be a private navigation
protocol. In the real world, the parties interact
according to the protocol specification $\pi$. Let $\mc{E}$ be the environment
and let $\mc{A}$ be an adversary that corrupts either the client or the server. The
protocol execution in the real world proceeds as follows:
\begin{enumerate}
\item \textbf{Inputs:} The environment $\mc{E}$ chooses a source-destination
pair $s, t \in V$ for the client and compressed next-hop routing matrices
$A^{(\dne)}, B^{(\dne)}, A^{(\dnw)}, B^{(\dnw)} \in \Z^{\vorig \times \vcomp}$
for the server.
The bit-length of all entries in the matrix products
$A^{(\dne)} \cdot (B^{(\dne)})^T$ and $A^{(\dnw)} \cdot (B^{(\dnw)})^T$
must be at most $\tau$. Finally, the environment gives the input of the corrupted
party to the adversary.

\item \textbf{Protocol Execution:} The parties begin executing the
protocol. All honest parties behave according to the protocol specification.
The adversary $\mc{A}$ has full control over the behavior of the corrupted party
and sees all messages received by the corrupted party.

\item \textbf{Output:} The honest party computes and gives its output to the
environment $\mc{E}$. The adversary computes a function of its view of the
protocol execution and gives it to $\mc{E}$.
\end{enumerate}
At the conclusion of the protocol execution, the environment $\mc{E}$ outputs
a bit $b \in \zo$. Let $\execr_{\pi, \mc{A}, \mc{E}}(\lambda)$ be the random variable
corresponding to the value of this bit. 
\end{definition}

\begin{definition}[Ideal Model of Execution]
\label{def:ideal-model}
\normalfont
In the ideal world, the client and server have access to a trusted
party $\mc{T}$ that computes the shortest paths functionality $f$.

\begin{enumerate}
\item \textbf{Inputs:} Same as in the real model of execution.

\item \textbf{Submission to Trusted Party:} If a party is honest, it
gives its input to the trusted party. If a party is corrupt, then
it can send any input of its choosing to $\mc{T}$, as directed by the
adversary $\mc{A}$.

\item \textbf{Trusted Computation:} From the next-hop routing matrices,
the trusted party computes the first $R$ hops on the shortest path
from $s$ to $t$: $s = v_0 , v_1, \ldots, v_R$. If $v_i = t$ for some
$i < R$, then the trusted party sets $v_{i+1}, \ldots, v_R$ to $\perp$.
If the next hop in the path at $v_i$ for some $i$ refers to a node not in $\mc{G}$,
then the trusted party sets $v_{i+1}, \ldots, v_R$ to $\perp$. The trusted
party sends the path $v_0, \ldots, v_R$ to the client. The server receives
no output.

\item \textbf{Output:} An honest party gives the sequence of messages
(possibly empty) it received from $\mc{T}$ to $\mc{E}$. The adversary
computes a function of its view of the protocol execution and gives it
to $\mc{E}$.
\end{enumerate}
At the conclusion of the protocol execution, the environment $\mc{E}$ outputs
a bit $b \in \zo$. Let $\execi_{f, \mc{A}, \mc{E}}(\lambda)$ be the random variable
corresponding to the value of this bit. 
\end{definition}

To state our security theorems, we now define the environment's distinguishing advantage.
Informally, we will say that a protocol is secure if no polynomial-size
environment is able to distinguish the real execution from the ideal execution with
non-negligible probability.

\begin{definition}[Distinguishing Advantage --- Security]
\label{def:distinguishing-advantage}
\normalfont
  Let $\pi$ be a private navigation protocol, and let $f$ be the shortest path functionality.
  Fix an adversary $\mc{A}$, simulator $\mc{S}$, and an environment $\mc{E}$.
  \iftoggle{fullversion}{
    Then, the distinguishing advantage $\advsec_{\pi, f, \mc{A}, \mc{S}, \mc{E}}$
    of $\mc{E}$ in the security game is given by
    \[ \advsec_{\pi, f, \mc{A}, \mc{S}, \mc{E}}(\lambda) = 
       \abs{\Pr[\execr_{\pi, \mc{A}, \mc{E}}(\lambda) = 0] - 
            \Pr[\execi_{f, \mc{A}, \mc{E}}(\lambda) = 0]}. \]
  }{
    The distinguishing advantage $\advsec_{\pi, f, \mc{A}, \mc{S}, \mc{E}}(\lambda)$
    of $\mc{E}$ in the security game is given by
      \[ \abs{\Pr[\execr_{\pi, \mc{A}, \mc{E}}(\lambda) = 0] - 
            \Pr[\execi_{f, \mc{A}, \mc{E}}(\lambda) = 0]}. \]
  }
\end{definition}

We will also work with a weaker notion of \emph{privacy} against a malicious adversary.
Informally, we say that the protocol is private if an adversary is unable to learn anything
about the inputs of the other party beyond what is explicitly leaked by the
inputs and outputs of the computation. To formalize this notion, we use the conventions in
\cite{IKKLP11} and define the distinguishing advantage in the privacy game.

\begin{definition}[Distinguishing Advantage --- Privacy]
\normalfont
  Let $\pi$ be a private navigation protocol, and let $f$ be the shortest path functionality.
  Fix an adversary $\mc{A}$, simulator $\mc{S}$, and an environment $\mc{E}$.
  Define $\execr'_{\pi,\mc{A},\mc{E}}(\lambda)$ exactly as
  $\execr_{\pi, \mc{A}, \mc{E}}(\lambda)$ (Definition~\ref{def:real-model}),
  except in the final step of the protocol execution, the environment
  only receives the adversary's output (and \emph{not} the honest party's output).
  Define $\execi'_{f,\mc{S},\mc{E}}(\lambda)$ analogously.
  \iftoggle{fullversion}{
    Then, the distinguishing advantage $\advpriv_{\pi, f, \mc{A}, \mc{S}, \mc{E}}$
    of $\mc{E}$ in the privacy game is given by
    \[ \advpriv_{\pi, f, \mc{A}, \mc{S}, \mc{E}}(\lambda) = 
       \abs{\Pr[\execr'_{\pi, \mc{A}, \mc{E}}(\lambda) = 0] - 
            \Pr[\execi'_{f, \mc{A}, \mc{E}}(\lambda) = 0]}. \]
  }{
    The distinguishing advantage $\advpriv_{\pi, f, \mc{A}, \mc{S}, \mc{E}}(\lambda)$
    of $\mc{E}$ in the privacy game is given by
    \[ \abs{\Pr[\execr'_{\pi, \mc{A}, \mc{E}}(\lambda) = 0] - 
          \Pr[\execi'_{f, \mc{A}, \mc{E}}(\lambda) = 0]}. \]
  }
\end{definition}

\subsection{Security Theorems}

The first requirement is that our protocol provides security against
a malicious client. This captures the notion that a malicious client
does not learn anything more about the server's routing information beyond the shortest
path between its requested endpoints and the publicly available information. In our setting,
we allow a privacy-performance trade-off where the client has a small probability
($R \cdot 2^{-\mu}$, where $\mu$ is the statistical security parameter)
of learning additional information about the routing information.
Since the order $p$ of the finite field must satisfy
$p > 2^{\tau + \mu + 1}$, using larger finite fields will decrease the failure probability, but
at the expense of performance.
In our experiments, $R \cdot 2^{-\mu} \approx 2^{-30}$.
We now state the formal security guarantee, but defer its formal proof to
\iftoggle{fullversion}{Appendix~\ref{app:malicious-client-security-proof}}
{the extended version of this paper}.

\begin{theorem}[Security Against a Malicious Client]
\label{thm:malicious-client-security}
Let $\pi$ be the protocol in Figure~\ref{fig:protocol} instantiated with
a CPA-secure encryption scheme $(\ms{Enc},\ms{Dec})$, a secure PRF $F$,
and an OT scheme secure against a malicious client. Let $\lambda, \mu$ be the
security parameter and statistical security parameter, respectively.
Let $f$ be the ideal
shortest-paths functionality. Then, for all $\ppt$ adversaries
$\mc{A}$, there exists a $\ppt$ adversary $\mc{S}$ such that for every
polynomial-size circuit family $\mc{E} = \set{\mc{E}}_\lambda$,
\[ \advsec_{\pi, f, \mc{A}, \mc{S}, \mc{E}}(\lambda) \le \negl(\lambda) + R \cdot 2^{-\mu}, \]
where $\negl(\lambda)$ denotes a negligible function in $\lambda$.
\end{theorem}

In addition to security against a malicious client, we require our protocol to
provide privacy against a malicious server. In other words, while a malicious server
might be able to cause the client to receive an invalid path, it still
cannot learn any information about the client's source or destination. We formalize
this in the following theorem.
\iftoggle{fullversion}{We defer the formal proof to
Appendix~\ref{app:malicious-server-privacy-proof}.}{}

\begin{theorem}[Privacy Against a Malicious Server]
\label{thm:malicious-server-privacy}
Let $\pi$ be the protocol in Figure~\ref{fig:protocol} instantiated with PIR and OT
primitives that provide privacy against a malicious server. Let $\lambda$ be a security
parameter and let $f$ be the ideal
shortest-paths functionality. Then, for all $\ppt$
adversaries $\mc{A}$, there exists a $\ppt$ adversary $\mc{S}$ such that
for every polynomial-size circuit family $\mc{E} = \set{\mc{E}}_\lambda$,
\[ \advpriv_{\pi, f, \mc{A}, \mc{S}, \mc{E}}(\lambda) \le \negl(\lambda), \]
where $\negl(\lambda)$ denotes a negligible function in $\lambda$.
\end{theorem}
\iftoggle{fullversion}{}{
  \begin{proof}[Proof (Sketch)]
  We give a sketch of the proof, and defer the full argument to the extended
  version of this paper. We argue that the server's view
  of the protocol execution can be simulated independently of the client's input.
  At a high-level, this follows from the fact that the server's view in the protocol
  execution consists only of its view in OT and PIR protocols. By assumption, privacy
  of the OT and PIR protocols implies the existence of a simulator that can simulate
  the server's view of the OT or PIR protocol independently of the client's input.
  Thus, for any adversarial server $\mc{A}$ in the real-world,
  we can construct a simulator $\mc{S}$ that is able to simulate a view computationally
  indistinguishable from that of $\mc{A}$ in the real protocol (by invoking the 
  underlying PIR and OT simulators for each sub-protocol).
  \end{proof}
}

\section{Experiments}
\label{sec:implementation}

In this section, we describe our implementation of the private routing protocol
from Figure~\ref{fig:protocol}.
Then, we describe our procedure for preprocessing and compressing actual road networks
for major cities taken from OpenStreetMap~\cite{OpenStreetMap}. Finally, we give
concrete performance benchmarks for our preprocessing and compression
pipeline as well as our private routing protocol on actual road networks.

\subsection{Protocol Implementation}
\label{sec:protocol-implementation}

To evaluate the performance of the protocol in Figure~\ref{fig:protocol},
we implemented the complete protocol in C++. In this section, we describe the building blocks
of our implementation. For each
primitive, we choose the parameters to guarantee a minimum of 80 bits of security. The complete
protocol implementation contains approximately 4000 lines of code.

\para{PIR.} We implemented the (recursive) PIR
protocol based on additive homomorphic
encryption from~\cite{KO97,OS07}. We instantiate the additive
homomorphic encryption scheme
with Paillier's cryptosystem~\cite{Pai99}, and use NTL~\cite{NTL}
over GMP~\cite{GMP} to implement the necessary modular
arithmetic. We use a 1024-bit RSA modulus for the plaintext space in the Paillier cryptosystem, which provides
80 bits of security. We use two levels of recursion in the PIR protocol, so the communication
scales as $O(\sqrt[3]{n})$ for an $n$-record database.

\para{OT.} We instantiate the OT protocol with the protocol from
\cite[\S 7.3]{HL10} which provides security against malicious clients and privacy
against malicious servers. This protocol is a direct generalization
of the Naor-Pinkas OT protocol~\cite{NP01}
based on the decisional Diffie-Hellman (DDH) assumption. Security
against a malicious client is enforced by having the client
include a zero-knowledge proof of knowledge (specifically, a Schnorr proof~\cite{Sch89})
with its OT request. To decrease the number
of rounds of communication, we apply the Fiat-Shamir heuristic~\cite{FS86} to
transform the interactive proof of knowledge into a non-interactive one
by working in the random oracle model. We instantiate the random
oracle with the hash function SHA-256.
For improved performance, we implement the Naor-Pinkas OT protocol
over the 256-bit elliptic curve group \texttt{numsp256d1} from~\cite{BCLN14a}.
We use the MSR-ECC~\cite{BCLN14a} library for the implementation of the
underlying elliptic curve operations. The 256-bit curve provides
128 bits of security.

\para{Arithmetic and Yao's circuits.} We implement our arithmetic circuits over
the finite field $\F_p$ where $p = 2^{61} - 1$ is a Mersenne prime. Then,
reductions modulo~$p$ can be performed using just two $p$-bit additions.
We use NTL~\cite{NTL} over GMP~\cite{GMP} for the finite field arithmetic.

For the garbled circuit implementation, we use
JustGarble~\cite{BHKR13} with the ``free XOR''~\cite{KS08} and
row-reduction optimizations~\cite{PSSW09}. We set the parameters of the
garbling framework to obtain 80-bits of security.
We use the optimized addition, comparison, and multiplexer circuits from~\cite{KSS11}
to implement the neighbor-computation function shown in Figure~\ref{fig:neighbor-computation}.
For multiplication, we implement the basic ``school method.''

\para{Record encryption and PRF.} We instantiate the CPA-secure encryption scheme
in Figure~\ref{fig:protocol} with AES in counter mode.  We also instantiate the PRF
used for deriving the neighbor keys (Step~\ref{enum:neighbor-key-derivation} in
Figure~\ref{fig:protocol}) with AES. We use the implementation of AES
from OpenSSL~\cite{OpenSSL}.

\subsection{Preprocessing and Map Compression.}
\label{sec:preprocessing-implementation}
We extract the street maps for four major cities
(San Francisco, Washington,~D.C., Dallas, and Los Angeles) from OpenStreetMap~\cite{OpenStreetMap}.
For each city, we take its most important roadways based on annotations in OpenStreetMap,
and construct the resulting graph $\mc{G}$.
Specifically, we introduce a node for each
street intersection in the city and an edge for each roadway. We assign edge weights
based on the estimated time needed to traverse the associated road segment (computed by taking the
length of the segment and dividing by the approximated speed limit along the segment).
Using the procedure described in Section~\ref{sec:graph-compression},
we preprocess the graph to have out-degree at most $4$. We then
associate each edge of $\mc{G}$ with a cardinal direction
by solving the assignment problem from Section~\ref{sec:graph-compression}.
We use Stachniss' implementation~\cite{Sta04}
of the Hungarian method~\cite{KY55} to solve this assignment problem.

Given the graph $\mc{G}$ corresponding to the road network for each city, we run
Dijkstra's algorithm~\cite{Dij59} on each node $s$ in $\mc{G}$ to
compute the shortest path between all pairs of nodes. Then, using the all-pairs shortest
paths information, we construct the next-hop routing matrices
$(\matrixind{M}\dne, \matrixind{M}\dnw)$ for $\mc{G}$. We remark that we can
substitute any all-pairs shortest path algorithm for Dijkstra's in this step.
The underlying principle we exploit in the construction of our protocol
is the fact that next-hop routing matrices for road networks have a simple compressible structure
amenable to cryptography.

Finally, we implement the optimization-based compression approach described in
Section~\ref{sec:graph-compression} to compress the next-hop routing matrices
$\matrixind{M}\dne$ and $\matrixind{M}\dnw$. We minimize the objective function
from Eq.~\eqref{eq:lbfgs-objective} with the loss function set to
the modified Huber hinge loss from Eq.~\eqref{eq:huber-loss}.
Because of the highly 
parallelizable nature of the objective function,
we write specialized CUDA kernels to evaluate the objective function and its derivative on the GPU.
In our experiments, we use the LBFGS optimization algorithm~\cite{BLNZ95}
from the Python scientific computation libraries NumPy and SciPy~\cite{ADHHO01}
to solve the optimization problem.

\subsection{Experiments}
\label{sec:experiments}

\begin{table}[t]
  \begin{center}
  \begin{tabular}{l | c | c | c}
    City & $n$ & Preprocessing Time (s) & Compression Time (s) \\ \hline
    San Francisco & 1830 & 0.625 & 97.500 \\
    Washington,~D.C. & 2490 & 1.138 & 142.431 \\
    Dallas & 4993 & 4.419 & 278.296 \\
    Los Angeles & 7010 & 9.188 & 503.007
  \end{tabular}
  \end{center}
  \caption{Average time to preprocess and compress the next-hop routing matrices for different
  networks. The second column gives the number of nodes $n$ in each city's road network.
  The preprocessing time column gives the average time needed to orient the
  edges, compute all-pairs shortest paths, and construct
  the next-hop routing matrix for the network. The compression
  time column gives the average time needed to compress the $\dne$ or $\dnw$ component
  of the next-hop routing matrices.
  \label{tab:preproc-times}}
\end{table}

\para{Graph preprocessing and compression.}
We first measure the time needed to preprocess and compress
the next-hop routing matrices for several road networks.
The preprocessing time includes the time needed to orient
the edges, compute all-pairs shortest paths, and
construct the next-hop routing matrix for the network
(as described in Section~\ref{sec:preprocessing-implementation}).

We also measure the time needed to compress the resulting next-hop routing matrices
for the different networks.
Recall that our compression method takes a matrix $M \in \set{-1,1}^{\vorig \times \vorig}$ and
produces two matrices $A,B \in \Z^{\vorig \times \vcomp}$ such that $\sign(AB^T)$ is a good
approximation of $M$. Since the modified Huber hinge loss (Eq.~\ref{eq:huber-loss})
is an upper bound
on the \mbox{0-1} loss function $\ell(x,t) = \mathbf{1}\set{\sign(x) = t}$,
when the objective value $J(A,B)$ is less than $1$
(where $J(A,B)$ is the objective function in
Eq.~\ref{eq:lbfgs-objective}), we have $\sign(AB^T) = M$, i.e., the matrices
$A,B$ \emph{perfectly reconstruct} $M$. The parameter $\vcomp$ is the number of columns in the matrices $A$
and $B$. Because our objective function is non-convex in the variables
$A$ and $B$, LBFGS is neither guaranteed to find the globally optimal solution, nor even
to converge in a reasonable number of iterations.
As a heuristic for deciding whether a candidate value of $\vcomp$ admits a feasible solution that
perfectly reconstructs $\matrixind{M}\dne$ and $\matrixind{M}\dnw$,
we run up to 5000 iterations of LBFGS and check whether the resulting solution gives a perfect
reconstruction of $M$. To determine the most compact representation,
we search over a range of possible values for $\vcomp$, and
choose the smallest value $\vcomp$ that yields a perfect reconstruction
of $M$.

We apply our compression method to the routing matrices for road networks 
from four cities of varying size. Then, we compare
the size of the original matrix $M$ to the size of its compressed representation $A,B$. The number
of bits needed to represent $A,B$ is determined by two factors: the number of
columns $\vcomp$ in each matrix $A,B$ and the precision $\nu$ (measured in number of bits)
needed to represent each
entry in $A,B$. Recall that the optimization procedure outputs two \emph{real-valued}
matrices such that $\sign(AB^T) = M$. To obtain a representation over the integers
(as required by the arithmetic circuits),
we scale the entries of $A,B$ by a constant factor
and round each of the resulting entries to the nearest integer. The precision
$\nu$ is the number of bits needed to represent each integer
component of $A,B$ after rescaling. 
We choose the smallest scaling factor such that the rescaled matrices
perfectly reconstruct the routing matrix $M$.

We run the preprocessing and compression experiments on a machine running Ubuntu 14.04 with
an 8-core 2.3 GHz Intel Core i7 CPU, 16 GB of RAM, and an Nvidia GeForce GT
750M GPU. The preprocessing and compression times for the different networks
are summarized in Table~\ref{tab:preproc-times}. A description of the compressed
representation of the routing matrices for each network is
given in Table~\ref{tab:city-parameters}.

In Figure~\ref{fig:compression}, we show the time needed to compress a single component
of the next-hop routing matrix, as well as the resulting compression factor, for subgraphs
of the road network for Los Angeles. The compression is quite effective,
and the achievable compression factor increases with the size of the network.
Moreover, even though the sizes of the next-hop routing matrices increase quadratically in
the number of nodes in the graph, the optimization time remains modest.
For graphs with 7000 nodes (and 350,000 optimization variables), finding a compact representation
that perfectly reconstructs the next-hop matrix completes in under 10 minutes. Since we
compress both the $\dne$ and $\dnw$ components of the routing matrix, the total time to both
preprocess and compress the shortest path information for the full city of Los Angeles is just
over 15 minutes. Lastly, we note that the preprocessing time for each network is small:
orienting the edges and computing all-pairs shortest paths via Dijkstra's algorithm
completes in under 10 seconds.

\begin{table}[t]
  \begin{center}
  \begin{tabular}{l | c | c | c | c | c}
  City & $n$ & $\vcomp$ & $\nu$ & $\tau$  & Compression Factor \\ \hline
  San Francisco & 1830 & 12 & 10 & 20 & 7.63 \\
  Washington,~D.C. & 2490 & 14 & 10 & 19 & 8.89 \\
  Dallas & 4993 & 19 & 12 & 23 & 10.95 \\
  Los Angeles & 7010 & 26 & 12 & 24 & 11.23
  \end{tabular}
  \end{center}
  \caption{Parameters for the compressed representation of the road networks for each city: $n$ is
  the number of nodes in each network, $\vcomp$ and $\nu$ are the number of columns and the precision, respectively, in the
  routing matrices $\matrixind{A}{\dne}$, $\matrixind{B}{\dne}, \matrixind{A}{\dnw}, \matrixind{B}{\dnw}$
  of the compressed representation, and $\tau$ is the maximum number of bits
  needed to represent an element in the products $\matrixind{A}{\dne}(\matrixind{B}{\dne})^T$ and
  $\matrixind{A}{\dnw}(\matrixind{B}{\dnw})^T$. The last
  column gives the compression factor attained for each network (ratio of size of uncompressed
  representation to size of compressed representation).}
  \label{tab:city-parameters}
\end{table}

\begin{figure}[t]
\centering
\iftoggle{fullversion}{\includegraphics[height=6cm]{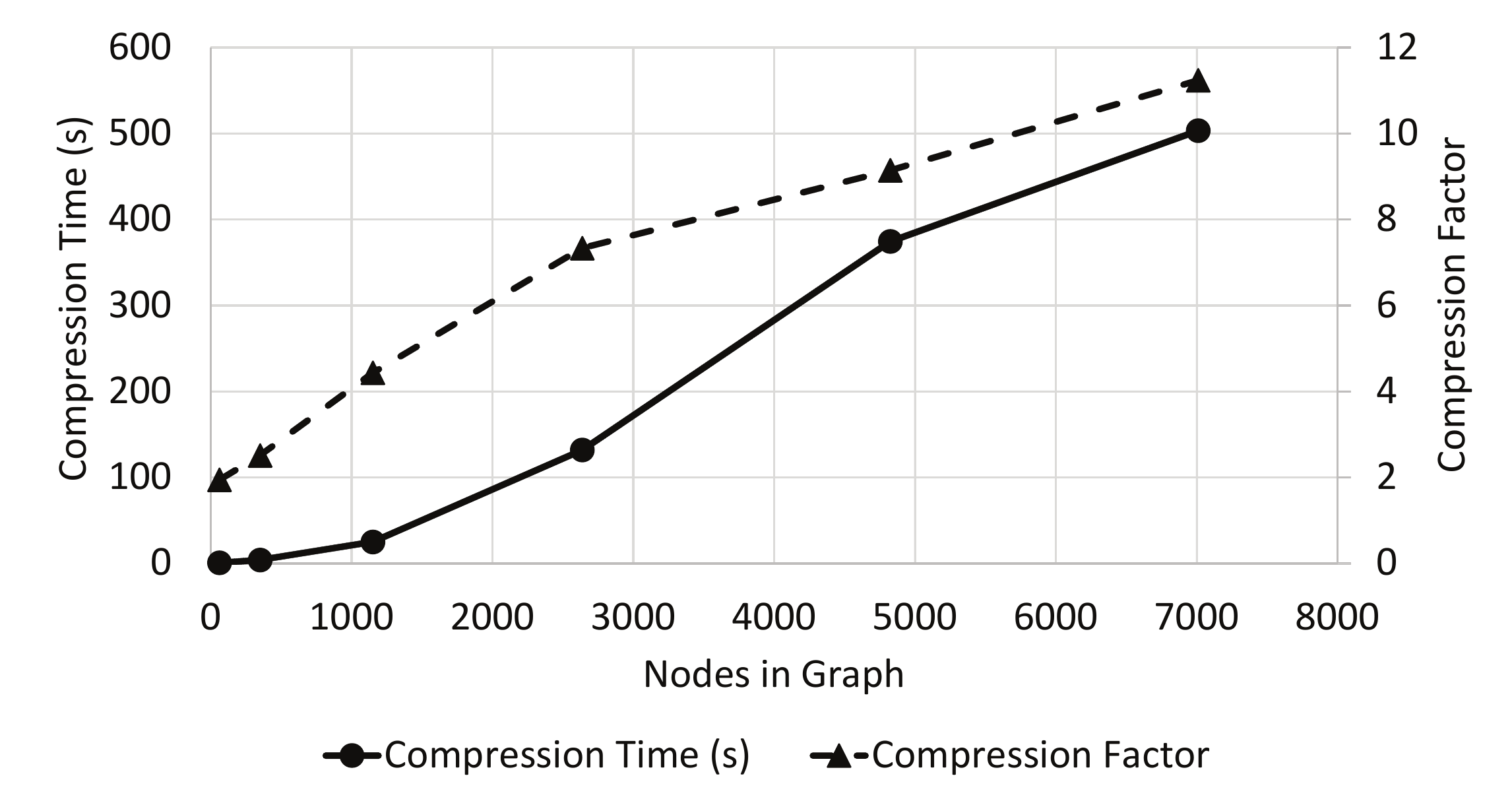}}
                      {\includegraphics[width=8cm]{figures/compression-all}}
\caption{Average time needed to compress the next-hop routing matrix and the resulting
compression factor for networks constructed from subgraphs of the road network of Los Angeles.}
\label{fig:compression}
\end{figure}

\begin{table*}[t]
  \small
  \begin{center}
  \begin{tabular}{l | c | c | c c c | c | c c | c c}
  \multirow{2}{*}{City} & Total Time (s) & \multicolumn{4}{c|}{Client Computation (s)} & 
    \multicolumn{3}{c|}{Server Computation (s)} & \multicolumn{2}{c}{Bandwidth (KB)} \\ \cline{3-11}
  & (Single Round) & Total & PIR & OT & GC & Total & PIR & OT & Upload & Download \\ \hline
  San Francisco   & $1.44 \pm 0.16$ & 0.35 & 0.31 & 0.02 & 0.02 & 0.88 & 0.80 & 0.08 & 51.74 & 36.50 \\
  Washington,~D.C. & $1.64 \pm 0.13$ & 0.38 & 0.34 & 0.02 & 0.02 & 1.07 & 1.00 & 0.08 & 52.49 & 37.51 \\
  Dallas          & $2.91 \pm 0.19$ & 0.45 & 0.41 & 0.02 & 0.02 & 2.19 & 2.11 & 0.08 & 55.50 & 39.52 \\
  Los Angeles     & $4.75 \pm 0.22$ & 0.55 & 0.51 & 0.02 & 0.02 & 3.70 & 3.62 & 0.08 & 57.01 & 43.53
  \end{tabular}
  \end{center}
  \caption{Performance benchmarks (averaged over at least 90 iterations)
  for a single round of the private routing protocol described in
  Figure~\ref{fig:protocol} on road networks for different cities.
  The ``Total Time'' column gives the average runtime and standard deviation
  for a single round of the protocol (including network communication times between a client and
  a server on Amazon EC2). The PIR,
  OT, and GC columns in the table refer to the time to perform the PIR for the affine
  encodings, the time to perform the OT for the garbled circuit encodings, and the time needed
  to evaluate the garbled circuit, respectively. The bandwidth measurements are taken with
  respect to the client (``upload'' refers to communication from the client to the server)}
  \label{tab:city-routing-computation}
\end{table*}

\begin{table*}[t]
  \small
  \begin{center}
  \begin{tabular}{l | c | c | c | c | c | c | c}
  \multirow{2}{*}{City} & \multirow{2}{*}{$R$} & \multicolumn{2}{c|}{Offline Setup} & \multicolumn{2}{c|}{Online Setup} & Total Online & Total Online \\ \cline{3-6}
  & & Time (s) & Band. (MB) & Time (s) & Band. (MB) & Time (s) & Bandwidth (MB) \\ \hline
  San Francisco & 97 & 0.135 & 49.08 & 0.73 & 0.021 & 140.39 & 8.38 \\
  Washington,~D.C. & 120 & 0.170 & 60.72 & 0.76 & 0.023 & 197.48 & 10.57 \\
  Dallas & 126 & 0.174 & 63.76 & 0.92 & 0.027 & 371.44 & 11.72 \\
  Los Angeles & 165 & 0.223 & 83.49 & 1.00 & 0.028 & 784.34 & 16.23 \\
  \end{tabular}
  \end{center}
  \caption{End-to-end performance benchmarks for the private routing protocol in Figure~\ref{fig:protocol}
  on road networks for different cities. For each network, the number of iterations $R$ is set to the
  maximum length of the shortest path between two nodes in the network. The offline computation
  refers to the server preparation and garbling of the $R$ circuits for evaluating the neighbor-computation
  function from Figure~\ref{fig:neighbor-computation}. The offline computation time
  just includes the computational cost and does \emph{not} include the garbled circuit download time.
  The online setup measurements correspond to computation
  and communication in the ``Setup'' phase of the protocol in Figure~\ref{fig:protocol}. 
  The ``Total Online Time'' and ``Total Online Bandwidth'' columns
  give the total end-to-end time (including network communication) and total communication between the client
  and server in the \emph{online} phase (navigation component) of the protocol.
  \label{tab:city-end-to-end}
  }
\end{table*}

\para{Performance on road networks.}
Next, we measure the run-time and bandwidth requirements of our private
routing protocol from Figure~\ref{fig:protocol}.
Table~\ref{tab:city-parameters} gives the
number of columns $\vcomp$, and the precision $\nu$ of the compressed representation of the
networks for the different cities. In addition, we also compute
the maximum number of bits $\tau$ needed to encode an element in 
the products
$\matrixind{A}{\dne}(\matrixind{B}{\dne})^T$ and
$\matrixind{A}{\dnw}(\matrixind{B}{\dnw})^T$.
From Theorem~\ref{thm:malicious-client-security}, a malicious client can
successfully cheat with probability at most $R \cdot 2^{-\mu}$, where
$\mu$ is the statistical security parameter, and $R$ is the total number of
rounds in the protocol. For each network in our experiments, we set the number of rounds
$R$ to be the maximum length over all shortest paths between any
source-destination pair in the network. This ensures both correctness (at the end
of the protocol execution, the client obtains the complete shortest path from its
source to its destination) as well as hides the length of the requested
shortest path from the server (since the number of rounds in the protocol is independent of the
client's input). Next, recall the relation
between $\mu$ and the order $p$ of the finite field for the affine encodings:
$p > 2^{\tau + \mu + 1}$.  In our experiments, we fix $p = 2^{61}-1$, and 
$R$ is at most $2^8 = 256$. This choice of parameters translates to $\mu$ ranging from $36$ to
$41$, or analogously, a failure probability of
$2^{-33}$ for the smaller networks to $2^{-28}$
for larger networks. Using larger fields will reduce this probability, but
at the expense of performance.

To reduce the communication in each round of the protocol in Figure~\ref{fig:protocol}, 
we note that it is not necessary for the server to prepare and send a garbled circuit to the client
on each round of the routing protocol. Since the neighbor-computation circuit is independent
of the state of the protocol execution, the circuits can be generated and stored long before the protocol execution
begins. Thus, in an \emph{offline} phase, the server can prepare and transmit to the client
a large number of garbled circuits.
During the online protocol execution, on the $\ord{r}$ round, the server just sends
the encodings corresponding to its input to the client;
it does \emph{not} send the description of the garbled
circuit. This significantly reduces the communication cost of each round of the online protocol.
We note that even if the routing matrices
$\matrixind{A}{\dne}, \matrixind{B}{\dne}, \matrixind{A}{\dnw}, \matrixind{B}{\dnw}$
changed (for instance, due to updates in traffic or weather conditions in the network)
during the protocol execution, as long as the bound $\tau$ on the bit-length of
entries in the products
$\matrixind{A}{\dne} (\matrixind{B}{\dne})^T$ and
$\matrixind{A}{\dnw} (\matrixind{B}{\dnw})^T$
remain fixed, the client and server do \emph{not} have to redo this offline
setup phase. We describe our extension for supporting live updates to the routing
information in greater detail in Section~\ref{sec:extensions}.

We run the server on a compute-optimized Amazon EC2 instance (running Ubuntu 14.04
with a 32-core 2.7 GHz Intel Xeon E5-2680v2 processor and 60 GB of RAM) to model
the computing resources of a cloud-based map provider.
The throughput of our protocol is bounded by the PIR computation on the server's side.
We use up to 60 threads on the server for the PIR computation.
All other parts of our system are single-threaded. For the client,
we use a laptop running Ubuntu 14.04 with
a 2.3 GHz Intel Core i7 CPU and 16 GB of RAM. The connection speed on the client
is around 50 Mbps. Both client and server support the AES-NI instruction set, which
we leverage in our implementation.

First, we measure the cost of one round of the private navigation
protocol. We assume that the client has
already downloaded the garbled circuits prior to the start of the protocol.
Table~\ref{tab:city-routing-computation}
gives the total computation time and bandwidth per round of the routing protocol.
When measuring the total time, we measure the end-to-end
time on the client, which includes the time for the network round trips. Table~\ref{tab:city-routing-computation}
also gives a breakdown of
the computation in terms of each component of the protocol: PIR for the arithmetic circuit encodings,
OT for the garbled circuit encodings, and garbled-circuit evaluation for computing the next-hop.

We also measure the total end-to-end costs for a single shortest path query.
As noted earlier, we set the
number of rounds $R$ for each network to be the maximum length of any shortest path in the network.
Irrespective of the client's source or destination, the client and server always engage in exactly
$R$ rounds of the private navigation protocol. Table~\ref{tab:city-end-to-end} shows the total computation time and
bandwidth required to complete a shortest-path query in the different networks. In the end-to-end
benchmarks, we also measure the offline costs of the protocol, that is, the
time needed for the server to garble $R$ neighbor-computation circuits
and the amount of communication needed for the client to download the circuits. In addition, we
measure the computation and bandwidth needed in the online setup phase of the routing protocol
(Figure~\ref{fig:protocol}).

In our protocol, the online setup phase of the protocol consists of three rounds of interaction.
First, the server sends the client the public description of the map. Then the client OTs
for the source and destination keys for the first round of the protocol, which requires two rounds
of communication. As shown in Table~\ref{tab:city-end-to-end}, the online setup procedure
completes in at most a second and requires under 30 KB of communication in
our example networks.

Next, we consider the performance of each round of the protocol. From
Table~\ref{tab:city-routing-computation}, the most computationally 
intensive component of our protocol is computing the responses to the PIR
queries.
In our implementation, we use
a Paillier-based PIR, so the server must perform $O(n)$ modular
exponentiations on each round of the protocol.
While it is possible to use a less computationally-intensive PIR such as~\cite{MBFK14},
the bandwidth required is much higher in practice.
Nonetheless, our results demonstrate that
the performance of our protocol is within the realm of practicality for
real-time navigation in cities like San Francisco or
Washington,~D.C.

Lastly, we note that the offline costs are dominated essentially by communication. With
hardware support for AES, garbling 100 neighbor-computation circuits
on the server completes in just a quarter of a second. While garbling is fast, the
size of each garbled circuit is 518.2 KB. For city networks, we typically require 100-150
circuits for each shortest-path query; this corresponds to 50-100
MB of offline download prior to the start of the navigation protocol. The experimental
results, however, indicate that the number of garbled circuits required for an end-to-end
execution grows sublinearly in the size of the graph.
For example, the total number of rounds (and correspondingly, the number of required
garbled circuits) for a graph with 1800 nodes is just under 100, while for a graph
with almost four times more nodes, the number of rounds only increases by a factor
of 1.7. We also note that each
neighbor-computation circuit consists of just under 50,000 non-XOR gates.
In contrast, generic protocols
for private navigation that construct a garbled circuit for Dijkstra's algorithm
yield circuits that contain hundreds of millions
to tens of billions of non-XOR gates~\cite{CMTB13,CLT14,LWNHS15}.

Finally, we see that despite needing to pad the number of rounds to a worst-case
setting, the total cost of the protocol remains modest. For the city of Los Angeles,
which contains over 7000 nodes, a shortest-path query still completes in under 15 minutes
and requires just over 16 MB of total bandwidth. Moreover, since
the path is revealed edge-by-edge rather than only at the end of the
computation, the overall protocol is an efficient solution for fully-private
navigation.

\para{Comparison to other approaches for private navigation.} Many
protocols~\cite{DK05,LLLZ09,XSS14} have been proposed for private navigation, but most
of them rely on heuristics and do not provide strong security
guarantees~\cite{DK05,LLLZ09}, or
guarantee privacy only for the client's location, and not the server's routing
information~\cite{XSS14}. A different approach to fully-private navigation is to leverage
generic multiparty computation techniques~\cite{Yao86,GMW87}. For instance,
a generic protocol for private navigation is to construct a garbled circuit
for a shortest-path algorithm and apply Yao's protocol. This approach is quite
expensive since the entire graph structure must be embedded in the circuit.
For instance, Liu et al.~\cite{LWNHS15} demonstrate that a garbled circuit for
evaluating Dijkstra's algorithm on a graph with just 1024 nodes requires over
10~{\em billion} AND gates. The bandwidth needed to transmit a circuit of this
magnitude quickly grows to the order of GB. In contrast, even for a larger graph
with 1800 nodes, the total online and offline communication required by our protocol
is under 60 MB (and the online communication is under 10 MB).
Carter et al.~\cite{CMTB13,CLT14} describe methods for reducing the computational
and communicational cost of Yao's protocol by introducing a third
(non-colluding) party that facilitates the computation. Even with this
improvement, evaluating a single shortest path on a graph of 100 nodes still
requires over 10 minutes of computation. As a point of comparison, our protocols
complete in around 2-3 minutes for graphs that are 15-20 times larger.
Evidently, while the generic tools
are powerful, they do not currently yield a practical private navigation protocol.
We survey additional related work and techniques in Section~\ref{sec:related-work}.

\section{Extensions}
\label{sec:extensions}

In this section, we describe several extensions to our protocol: supporting
navigation between cities, handling updates to the routing information, and
updating the source node during the protocol execution (for instance, to
accommodate detours and wrong turns).

\para{Navigating between cities.} The most direct method for supporting
navigation across a multi-city region is to construct a network that spans
the entire region and run the protocol directly. However, since the server's
computation in the PIR protocols grows as $O(n \vcomp \log p)$, where $n$ is the number
of nodes in the graph, $\vcomp$ is the number of columns in the compressed
representation, and $p$ is the order of the finite field used for the affine
encodings, this can quickly become computationally infeasible for the server.

An alternative method that provides a performance-privacy trade-off is to
introduce a series of publicly-known waypoints for each city. For example,
suppose a user is navigating from somewhere in Los Angeles to somewhere in San
Diego. In this case, the user would first make a private routing request to
learn the fastest route from her current location to a waypoint in Los Angeles.
Once the user arrives at the waypoint in Los Angeles, she requests the fastest
route to a waypoint in San Diego. This second query is performed entirely in the
clear, so the user reveals to the server that she is traveling from Los Angeles
to San Diego. Once the user arrives at a waypoint in San Diego, she makes a
final private routing request to learn the fastest route to her destination. In
this solution, the server only obtains a macro-view of the user's location: it
learns only the user's source and destination cities, and no information about
the user's particular location within the city. As we have demonstrated, the
protocol in Figure~\ref{fig:protocol} is able to handle real-time navigation for
a single city; thus, using this method of waypoints, we can also apply our
protocol to navigation between cities with limited privacy loss.

\para{Live updates to routing information.} Routing information in road
networks is dynamic, and is influenced by current traffic conditions, weather
conditions, and other external factors. Ideally, the edges revealed in 
an iterative shortest-path protocol should always correspond to the shortest path to the
destination given the current network conditions. It is fairly straightforward
to allow for updates to the routing information in our protocol. Specifically, we
observe that the compressed routing matrices
$\matrixind{A}{\dne}, \matrixind{A}{\dnw}, \matrixind{B}{\dne}, \matrixind{B}{\dnw}$
need not be fixed for the duration of the protocol. As long as the total number
of columns $\vcomp$, the bound on the bit-length $\tau$ of the values in the matrix products
$\matrixind{A}{\dne} \cdot (\matrixind{B}{\dne})^T$ and
$\matrixind{A}{\dnw} \cdot (\matrixind{B}{\dnw})^T$, and the total number of rounds
$R$ in the protocol remain fixed, the server can use a different set of routing
matrices on each round of the protocol. Therefore, we can accommodate
live updates to the routing information during the protocol execution by simply
setting a conservative upper bound on the parameters $\vcomp, \tau, R$. Note that
we can always pad a routing matrix with fewer than $\vcomp$ columns to one with
exactly $\vcomp$ columns by adding columns where all entries are $0$. Since computing
the shortest path information for a city-wide network and compressing the resulting
routing matrices completes in just a few minutes, it is possible to ensure
accurate and up-to-date routing information in practice.

\para{Updating sources and destinations.} Typically, in navigation,
the user might take a detour or a wrong turn. While the protocol is designed
to constrain the client to learn a single contiguous route through the network,
it is possible to provide a functionality-privacy trade-off to accommodate deviations
from the actual shortest path. One method is to introduce an additional parameter
$K$, such that after every $K$ iterations of the protocol, the server chooses fresh
source keys for the next round of the protocol. After every $K$ rounds, the
client would also OT for a new source key.
Effectively, we are resetting the protocol every $K$ rounds and allowing the client to
choose a new source from which to navigate. Correspondingly, we would need to increase the
total number of rounds $R$ in order to support the potential for detours and wrong turns.
Though we cannot directly bound the number of rounds $R$, we can use a conservative estimate.
Of course, a dishonest client can now learn multiple sub-paths to its chosen destination, namely,
one sub-path each time it is allowed to choose a different source. In a similar manner,
we can support updates to the destination.

\section{Related Work}
\label{sec:related-work}

Numerous approaches have been
proposed for private shortest path computation~\cite{DK05,LLLZ09,MY12,Mou13,
XSS14,CMTB13,CLT14,BSA13,WNLCSSH14,KS14,LWNHS15}.
Early works such as~\cite{DK05,LLLZ09}
propose hiding the client's location by either providing approximate locations
to the server~\cite{DK05} or by having the client submit dummy sources and destinations
with each shortest path query~\cite{LLLZ09}. However, these approaches only provide limited
privacy for the client's location. Later works~\cite{Mou13,MY12,XSS14}
describe PIR-based solutions for hiding the client's location.
In~\cite{Mou13,MY12}, the client
first privately retrieves subregions of the graph that are relevant to its query~\cite{Mou13,MY12} and then
locally computes the shortest path over the subgraph. In~\cite{XSS14},
the client privately requests for columns of the next-hop
routing matrix to learn the next hop in the shortest path. While these methods
provide privacy for the client's location, they do not hide the server's
routing information.

There is also work on developing secure protocols
for other graph-theoretic problems and under different models~\cite{BS05,FG06}. For example,
Brickell and Shmatikov~\cite{BS05} consider a model where two parties hold a graph over a common
set of vertices, and the goal is to compute a function over their joint graphs.
Their protocols do not extend to navigation protocols where one party holds the full graph,
and only the client should learn the result of the computation. In~\cite{FG06}, the authors
describe protocols for parties who each hold a subset of a graph to privately reconstruct the
joint graph. Their methods are designed for social network analysis and do not directly apply to
private navigation.

Another line of work has focused on developing data-oblivious algorithms for shortest path
computation~\cite{BSA13} or combining shortest path algorithms such as Dijkstra's with
oblivious data structures or ORAM~\cite{WNLCSSH14,KS14}. In these methods, the routing
data is stored in an ORAM or an oblivious data structure on the server. The client
then executes the shortest-path algorithm on the server
to learn the path between its source and destination.
Since the pattern of memory accesses is hidden from the server, these approaches provide client privacy.
While these protocols can be efficient in practice, they
do not provide security against a malicious client 
trying to learn additional details about the routing information on the server.
Thus, for scenarios where the map data is proprietary (for instance, in the case of real-time
traffic routing), or when the routing information itself is sensitive (for instance,
when providing navigational assistance for a presidential motorcade or
coordinating troop movements in a military scenario~\cite{CMTB13,CLT14}),
the ORAM-based solutions do not provide sufficient security.

Also relevant are the works in secure multiparty computation (MPC)~\cite{Yao86,GMW87}.
While these methods can be successfully used to build private navigation
protocols~\cite{LWNHS15,CMTB13,CLT14}, they do not currently yield a practical
private navigation protocol. A more comprehensive comparison of our protocol to these
generic methods is provided at the end of Section~\ref{sec:experiments}.

There is also a vast literature on graph compression algorithms. For planar graphs,
there are multiple methods based on computing graph separators~\cite{LT79,BBK03,BBK04}.
Other methods based on coding schemes~\cite{HKL00} have also been proposed and shown
to achieve information-theoretically optimal encoding. While these algorithms
are often viable in practice, it is not straightforward to represent them compactly as
a Boolean or an arithmetic circuit. Thus, it is unclear how to
combine them with standard cryptographic primitives
to construct a private shortest path protocol.

Finally, there has also been work on developing
compact representations of graphs for answering \emph{approximate} distance
queries in graphs~\cite{TZ01}. These techniques have been successfully
applied for privacy-preserving
approximate distance computation on graphs~\cite{MKNK15}. However, these distance-oracle-based
methods only provide an estimate on the {\em length} of the shortest path, and do not give
a private navigation protocol.

\section{Conclusion}

In this work, we constructed an efficient protocol for privately computing
shortest paths for navigation. First, we developed a method for compressing the
next-hop matrices for road networks by
formulating the compression problem as that of finding a sign-preserving, 
low-rank matrix decomposition. Not only did this method yield a significant
compression, it also enabled an efficient cryptographic protocol for fully private
shortest-path computation in road networks. By combining affine encodings
with Yao's circuits, we obtained a fully-private navigation protocol
efficient enough to run at a city-scale.

\section*{Acknowledgments}
The authors would like to thank Dan Boneh, Roy Frostig, Hristo Paskov, and
Madeleine Udell for many helpful comments and discussions. While
conducting this work, authors David Wu and Joe Zimmerman were supported by NSF
Graduate Research Fellowships. This work was further supported by the
DARPA PROCEED research program. Opinions, findings and conclusions or
recommendations expressed in this material are those of the authors and do
not necessarily reflect the views of DARPA or NSF.

\iftoggle{fullversion}{
  \bibliographystyle{alpha} 
}{
  \bibliographystyle{IEEEtranS}
}
\bibliography{references}

\iftoggle{fullversion}{
  \appendix

\section{Security Proofs}
\label{app:security-proofs}

In this section, we show that the protocol in Figure~\ref{fig:protocol} 
securely computes the shortest paths functionality
in the presence of a malicious client, and provides privacy against a malicious server.
To simplify our proofs, we work in the OT-hybrid model where we assume
the parties have access to an ideal 1-out-of-$n$ OT functionality~\cite{Kil88}.
Specifically, in the real protocol, we replace every OT invocation with
an oracle call to a trusted party that implements the OT functionality:
the sender sends the database of records $(r_1, \ldots, r_n)$ to the trusted party
and the receiver sends an index $i \in [n]$ to the trusted party. The 
trusted party then gives the receiver the record $r_i$.
Security in the standard model then follows by instantiating the
ideal OT functionality with an OT protocol that provides security
against malicious clients~\cite{HL10} and privacy against malicious servers,
and then invoking the sequential composition theorem of~\cite{Can00}.

\subsection{Proof of Theorem~\ref{thm:malicious-server-privacy}}
\label{app:malicious-server-privacy-proof}
At a high level, privacy for the client's location follows from the fact that the server's view in the
protocol execution consists only of its view in the OT and PIR protocols. By assumption,
both the OT protocols and the PIR protocols provide privacy for the client's input, so the claim
follows. We now show this formally. As noted at the beginning of Appendix~\ref{app:security-proofs},
we work in the OT-hybrid model, where we replace each OT invocation with an oracle
call to an ideal OT functionality. First, we state the definition of privacy as
it applies to the PIR protocol.
\begin{definition}[Privacy for PIR]
  \label{def:pir-privacy}
  \normalfont
  Fix a security parameter $\lambda \in \N$, and let $\pi$ be a PIR protocol.
  Let $\mc{A}$ be a non-uniform
  $\ppt$ server for $\pi$. Let $\view_{\pi, \mc{A}}(1^\lambda, \mc{D}, i)$
  denote the view of adversary
  $\mc{A}$ in the PIR protocol on database $\mc{D} \in \zo^*$ (the server's input) and index $i \in \zo^*$ (the
  client's input). Then, $\pi$ is a \emph{private} PIR protocol
  if for all non-uniform $\ppt$ servers $\mc{A}$, databases $\mc{D} \in \zo^*$, indices 
  $i, i' \in \zo^*$ where $\abs{i} = \abs{i'}$, we have that
  \[ \view_{\pi, \mc{A}}(1^\lambda, \mc{D}, i) \appc \view_{\pi, \mc{A}}(1^\lambda, \mc{D}, i'). \]
\end{definition}

Let $\mc{A}$ be a malicious server for the private shortest paths protocol in
Figure~\ref{fig:protocol}. We construct an ideal-world simulator $\mc{S}$
such that the distribution of outputs of $\mc{A}$ in the real protocol
is computationally indistinguishable from the outputs of $\mc{S}$ in the ideal
world. This suffices to prove that the protocol in Figure~\ref{fig:protocol}
provides \emph{privacy} against a malicious server. The simulator
$\mc{S}$ begins by running $\mc{A}$. Below, we describe how $\mc{S}$ simulates
the view for $\mc{A}$ in the protocol execution.

\paragraph{Setup.} In the setup phase of the protocol, the simulator does 
nothing. This is the correct behavior because in the OT-hybrid model,
the adversary $\mc{A}$ does not receive any messages during the setup phase.

\paragraph{Round.} On each round of the protocol, the simulator $\mc{S}$
plays the role of the client in the PIR protocol and requests for record $0$
in both the source database and in the 
destination database. Again, since we are working in the OT-hybrid model,
these are the only messages adversary $\mc{A}$ obtains in the
real protocol.
At the end of the protocol execution, adversary $\mc{A}$ will output some
function of its view of the protocol execution. The simulator $\mc{S}$
echoes this output to the environment.

\paragraph{Correctness of the simulation.}
To conclude the proof, it suffices to show that the view $\mc{S}$ simulates
for $\mc{A}$ is computationally indistinguishable from the view $\mc{A}$
expects in the real protocol. This condition holds vacuously in the setup phase of
the protocol. Let $\viewi{r}{\mc{A}}$ be the adversary's view on the $\ord{r}$
round of the protocol. The view $\viewi{r}{\mc{A}}$ may be written as
$\viewi{r}{\mc{A}} = \set{\viewi{r}{\pir, \mc{A}}(1^\lambda, \mc{D}_\src, i_\src),
                          \viewi{r}{\pir, \mc{A}}(1^\lambda, \mc{D}_\dst, i_\dst)}$,
where $\dsrc, \ddst$ are the encoding databases $\mc{A}$ chooses in the real protocol,
and $i_\src, i_\dst$ are the indices of the records the client chooses in the real
protocol. By privacy of the PIR (Definition~\ref{def:pir-privacy}), it follows
that
\begin{align*}
\viewi{r}{\pir, \mc{A}}(1^\lambda, \mc{D}_\src, i_\src) &\appc 
\viewi{r}{\pir, \mc{A}}(1^\lambda, \mc{D}_\src, 0) \\
\viewi{r}{\pir, \mc{A}}(1^\lambda, \mc{D}_\dst, i_\dst) &\appc
\viewi{r}{\pir, \mc{A}}(1^\lambda, \mc{D}_\dst, 0).
\end{align*}
Since the request for the source encodings and for the destination encodings
constitute two independent instances of the PIR protocol, we
conclude that the view $\mc{S}$ simulates for $\mc{A}$ in each round of the
protocol is computationally indistinguishable from the view $\mc{A}$ expects in the
real protocol. Thus, the output of $\mc{S}$ in the ideal-world execution is computationally
indistinguishable from that of $\mc{A}$ in the real world. \qed

\subsection{Proof of Theorem~\ref{thm:malicious-client-security}}
\label{app:malicious-client-security-proof}

Before we prove Theorem~\ref{thm:malicious-client-security}, we describe the
simulatability requirement we require on the garbled circuit encodings 
used in the protocol in Figure~\ref{fig:protocol}.
Intuitively, we require that the garbled circuit encodings
can be entirely simulated given the output of the computation; that is,
the garbled circuit together with one set of encodings do not reveal any information
about the underlying inputs other than what is explicitly revealed by the output.
Bellare et al. formalize this notion in~\cite{BHR12}. Here, we give a simplified
definition adapted from~\cite{GKPVZ13} and specialized to the case of Yao's
garbling scheme~\cite{Yao86,LP09}.

\begin{definition}[Yao Garbling Scheme]
  \label{def:yao-garb-scheme}
  \normalfont
  A Yao garbling scheme $\piyao$ for a family of circuits $\set{\mc{C}_n}_{n \in \mathbb{N}}$
  (where $\mc{C}_n$ is a set of Boolean circuits on $n$ input bits) consists of three algorithms
  $(\yaogarble, \yaoenc, \allowbreak \yaoeval)$ where
  \begin{itemize}
    \item $\yaogarble(1^\lambda, C)$ is a randomized algorithm that
    takes as input a security parameter $\lambda$ and a
    circuit $C \in \mc{C}_n$ for some $n$ and outputs a garbled circuit
    $\tilde C$ along with a secret key $\sk$ where
    $\sk = \set{L_i^0, L_i^1}_{i \in [n]}$ is a set containing $n$ pairs of encodings
    $L_i^0, L_i^1 \in \zo^*$.

    \item $\yaoenc(\sk, x)$ is a deterministic algorithm that takes the secret key
    $\sk = \set{L_i^0, L_i^1}_{i \in [n]}$ and an input $x = x_1 \cdots x_n \in \zo^n$,
    and outputs an encoding $\tilde x = \set{L_i^{x_i}}_{i \in [n]}$. Specifically,
    the encoding $\tilde x$ of $x$ is the subset of encodings in $\sk$
    associated with the bits of $x$.

    \item $\yaoeval(\tilde C, \tilde x)$ is a deterministic algorithm that takes
    a garbled circuit $\tilde C$ and a set of encodings
    $\tilde x = \set{L_i^{x_i}}_{i \in [n]}$ for some $x \in \zo^n$ and outputs a value $z$.
  \end{itemize}
\end{definition}

\begin{definition}[Correctness]
  \normalfont
  A Yao garbling scheme $\piyao = (\yaogarble, \yaoenc, \yaoeval)$
  for a family of circuits $\set{\mc{C}_n}_{n \in \mathbb{N}}$
  is correct if for all $n = \poly(\lambda)$, $C \in \mc{C}_n$, and $x \in \zo^n$, the following
  holds. Letting $(\tilde C, \sk) \gets \yaogarble(1^\lambda, C)$, then with overwhelming probability
  in $\lambda$, we have
  \[ \yaoeval(\tilde C, \yaoenc(\sk, x)) = C(x), \]
  where the probability is taken over the random coins used in $\yaogarble$.
\end{definition}

\begin{definition}[Input Privacy]
  \label{def:garbling-input-privacy}
  \normalfont
  A Yao garbling scheme $\piyao = (\yaogarble, \yaoenc, \yaoeval)$ for a family of circuits
  $\set{\mc{C}_n}_{n \in \mathbb{N}}$ is
  input-private if there exists a $\ppt$ simulator $\mc{S}_{\yao}$ such that for all
  $n = \poly(\lambda)$, $C \in \mc{C}_n$, $x \in \zo^n$, the following holds:
  \[ \set{(\tilde C, \sk) \gets \yaogarble(1^\lambda, C)\ ;\,(\tilde C, \yaoenc(\sk, x))}
    \appc \mc{S}_{\yao}(1^\lambda, C, C(x)). \]
\end{definition}

\begin{lemma}[{\cite{Yao86,LP09}}]
  \label{lem:yao-simulatable}
  Assuming one-way functions exist, there exists a Yao garbling scheme 
  (Definition~\ref{def:yao-garb-scheme}) that is input-private
  (Definition~\ref{def:garbling-input-privacy}).
\end{lemma}

\noindent Next, we note that the affine encodings from Section~\ref{sec:shortest-paths-protocol},
Eq.~\eqref{eq:arith-encoding} provide statistical privacy.
\begin{lemma}[{\cite[Lemma 5.1]{AIK14}, adapted}]
  \label{lem:affine-encoding-simulatable}
  Fix a finite field $\F_p$ of prime order $p$, and take $z_1, \ldots, z_\vcomp \in \F_p$.
  Define the function
  $f : \F_p^\vcomp \times \F_p^\vcomp \to \F_p$ where
  $f(x, y) = \langle x, y \rangle + \sum_{i \in [\vcomp]} z_i$. Let $\laffine_x$ and
  $\laffine_y$ be the affine encoding functions from Eq.~\eqref{eq:arith-encoding-vector}.
  Then, there exists a $\ppt$ simulator $\mc{S}_{\ms{ac}}$ such that for all
  $x, y \in \F_p^\vcomp$ and $z_1, \ldots, z_\vcomp \in \F_p$,
  \[ \mc{S}_\ms{ac}(f(x,y)) \equiv 
       \set{r \getsr \F_p^{3\vcomp}\ ;\,\left( \laffine_x(x; r), \laffine_y(y; r) \right)}. \]
\end{lemma}

\paragraph{Proof of Theorem~\ref{thm:malicious-client-security}.}
To show Theorem~\ref{thm:malicious-client-security}, we first define the following hybrid
experiments:
\begin{itemize}
  \item Hybrid $\hyb_0$: This is the real experiment (Definition~\ref{def:real-model}).

  \item Hybrid $\hyb_1$: Same as $\hyb_0$, except the protocol execution aborts if
  the client succeeds in making an ``inconsistent'' query (described below).

  \item Hybrid $\hyb_2$: This is the ideal experiment (Definition~\ref{def:ideal-model}).
\end{itemize}
Informally speaking, we say that the client succeeds in making an ``inconsistent'' query if
on some round $r$, it
requests the garbled circuit encodings for values $\hat z_\dne$ and $\hat z_\dnw$
that were not the outputs of the arithmetic circuit,
and yet, the client obtains a set of garbled circuit encodings where the
garbled circuit evaluation does not output $\perp$.
We now specify this property more precisely.

\paragraph{Specification of Hybrid~$\hyb_1$.} Let $s, t \in [n]$ be the source and
destination nodes the client sends to the ideal OT functionality in the
setup phase of the protocol in Figure~\ref{fig:protocol}. Let
$(s = v_0, v_1, \ldots, v_R)$ be the shortest path from $s$ to $t$ as defined
by the environment's choice of the next-hop routing matrices
$\matrixind{A}{\dne}, \matrixind{B}{\dne}, \matrixind{A}{\dnw}, \matrixind{B}{\dnw}$
in the protocol execution (Definition~\ref{def:real-model}).
The protocol execution in $\hyb_1$ proceeds identically to that in $\hyb_0$,
except the protocol execution halts (with output $\perp$) if the
following bad event occurs:

\vspace{1em}
\noindent \begin{minipage}{\textwidth}
\begin{framed}
  \noindent On a round $1 \le r \le R$, the client submits $\hat z_\dne, \hat z_\dnw$
  to the ideal OT functionality where either
  \[ \hat z_\dne \ne \alpha_\dne
     \langle \matrixind{A}{\dne}_{v_r}, \matrixind{B}{\dne}_t \rangle +
     \beta_\dne
   \quad \text{or} \quad
   \hat z_\dnw \ne \alpha_\dnw 
     \langle \matrixind{A}{\dnw}_{v_r}, \matrixind{B}{\dnw}_t \rangle +
     \beta_\dnw, \]
  and $\cunblind((\hat z_\dne, \gamma_\dne, \delta_\dne),
                 (\hat z_\dnw, \gamma_\dnw, \delta_\dnw),
                 k_\dne^0, k_\dne^1, k_\dnw^0, k_\dne^1, v_r, t) \ne \ \perp$, where
  all values other than $\hat z_\dne$, $\hat z_\dnw$, $v_r$, and $t$
  are the round-specific values chosen by the server on round $r$.
\end{framed}
\captionof{figure}{Abort event in hybrid $\hyb_1.$}
\label{fig:hyb-1-bad-event}
\end{minipage}
\vspace{1em}

Fix a security parameter $\lambda$.
Let $\pi$ be a private navigation protocol and let $f$ be the ideal shortest-path
functionality. For a client $\mc{A}$, a simulator $\mc{S}$ and an environment
$\mc{E}$, we define the following random variables:
\begin{itemize}
  \item $\hyb_0(\lambda, \pi, \mc{A}, \mc{E})$ is the output of experiment
  $\hyb_0$ with adversary $\mc{A}$ and environment $\mc{E}$.
  In particular,
  $\hyb_0(\lambda, \pi, \mc{A}, \mc{E}) = \execr_{\pi, \mc{A}, \mc{E}}(\lambda)$
  (Definition~\ref{def:real-model}).

  \item $\hyb_1(\lambda, \pi, \mc{A}, \mc{E})$ is the output of experiment
  $\hyb_1$ with adversary $\mc{A}$ and environment $\mc{E}$.

  \item $\hyb_2(\lambda, f, \mc{S}, \mc{E})$ is the output of experiment
  $\hyb_2$ with simulator $\mc{S}$ and environment $\mc{E}$. In particular,
  $\hyb_2(\lambda, f, \mc{S}, \mc{E}) = \execi_{f, \mc{S}, \mc{E}}(\lambda)$
  (Definition~\ref{def:ideal-model}).
\end{itemize}
To prove Theorem~\ref{thm:malicious-client-security}, we show the following
two claims.
\begin{claim}
  \label{claim:hyb-0-1}
  Let $\lambda, \mu$ be the security parameter and statistical security
  parameter, respectively. Let $\pi$ be the protocol in Figure~\ref{fig:protocol}
  instantiated with secure cryptographic primitives as described in 
  Theorem~\ref{thm:malicious-client-security}. Then,
  for all $\ppt$ adversaries $\mc{A}$, and every
  polynomial-size circuit family $\mc{E} = \set{\mc{E}}_\lambda$,
  \[ \abs{\Pr\left[ \hyb_0(\lambda, \pi, \mc{A}, \mc{E}) = 0 \right] - 
          \Pr\left[ \hyb_1(\lambda, \pi, \mc{A}, \mc{E}) = 0 \right]} \le
          \negl(\lambda) + R \cdot 2^{-\mu}. \]
\end{claim}

\begin{claim}
  \label{claim:hyb-1-2}
  Let $\lambda, \mu$ be the security parameter and statistical security
  parameter, respectively. 
  Let $\pi$ be the protocol in Figure~\ref{fig:protocol} instantiated with
  secure cryptographic primitives as described in 
  Theorem~\ref{thm:malicious-client-security}. 
  Let $f$ be the ideal shortest-paths functionality. Then, for all $\ppt$ adversaries
  $\mc{A}$, there exists a $\ppt$ adversary $\mc{S}$ such that for every
  polynomial-size circuit family $\mc{E} = \set{\mc{E}}_\lambda$,
  \[ \abs{\Pr\left[ \hyb_1(\lambda, \pi, \mc{A}, \mc{E}) = 0 \right] -
          \Pr\left[ \hyb_2(\lambda, f,   \mc{S}, \mc{E}) = 0 \right]} \le \negl(\lambda). \]
\end{claim}

\paragraph{Proof of Claim~\ref{claim:hyb-1-2}.}
We begin by showing Claim~\ref{claim:hyb-1-2}. Given a real-world adversary
$\mc{A}$, we construct an efficient ideal-world simulator $\mc{S}$
such that the distribution of outputs of any (possibly malicious) client $\mc{A}$
in $\hyb_1$ is computationally indistinguishable
from the distribution of outputs of an ideal-world simulator $\mc{S}$ in $\hyb_2$.
Since the server does not produce any output,
this will suffice to show that hybrids $\hyb_1$ and $\hyb_2$ are computationally
indistinguishable.

As stated in Section~\ref{sec:security-model} and Figure~\ref{fig:protocol},
we assume that the topology of the
graph $\mc{G} = (V, E)$, the number of columns $d$ in the compressed routing
matrices, the bound $\tau$ on the bit-length of the products of the compressed
matrices, and the total number of rounds $R$ are public and known to both
parties in the protocol execution.
Moreover, as described in Figure~\ref{fig:protocol},
we assume that
$p > 2^{\tau + \mu + 1}$ where $\mu \in \N$ is the statistical security parameter.
In particular, there exists an element
$\xi \in \F_p$ such that $\xi \not \in [-2^\tau, 2^\tau]$.

\paragraph{Specification of the simulator.} We begin by describing the simulator.
We let $\mc{T}$ denote the trusted party for the shortest path functionality as defined in 
the specification of the ideal model of execution from Section~\ref{sec:security-model}.
The simulator starts running adversary $\mc{A}$. We describe the behavior of the simulator in the setup
phase of the protocol as well as the behavior on each round of the routing protocol.

\paragraph{Setup.} In the setup phase of the protocol, the simulator $\mc{S}$
does the following:
\begin{enumerate}
\item As in the real protocol, the simulator first chooses independent symmetric
encryption keys $\matrixind{\bar k}{1}_{\src,i}, \bar k_{\dst,i} \getsr \zo^\ell$ for
all $i \in [n]$.

\item When $\mc{A}$ makes an OT request for an entry $s$ in the source key
database, the simulator replies with the key $\matrixind{\bar k}{1}_{\src, s}$. Recall that
in the ideal OT functionality, each party just sends its input to the ideal functionality,
and the ideal functionality sends the requested element to the client.

\item When $\mc{A}$ makes an OT request for an entry $t$ in the destination
key database, the simulator replies with the key $\bar k_{\dst,t}$ to $\mc{A}$.

\item Finally, the simulator sends $(s,t)$ to the trusted party $\mc{T}$.
The trusted party replies to the simulator with a path $(s = v_0, \ldots, v_R)$.
\end{enumerate}

\paragraph{Round.} Next, we describe the behavior of the simulator in each
round $1 \le r \le R$ of the protocol.

\begin{enumerate}
\item The simulator chooses $\bar z_\dne, \bar z_\dnw \getsr \F_p$. Then, $\mc{S}$ invokes
the simulator $\mc{S}_{\ms{ac}}$ (Lemma~\ref{lem:affine-encoding-simulatable})
on inputs $\bar z_\dne$ and $\bar z_\dnw$ to obtain four sets of encodings
$\blaffine_{\dne, x}, \blaffine_{\dne,y}, \blaffine_{\dnw, x}, \blaffine_{\dnw, y}$.

\item Let $\cunblind$ be a circuit computing the neighbor-computation function in
Figure~\ref{fig:neighbor-computation}. As in the real protocol, the simulator runs
Yao's garbling algorithm on $\cunblind$ to obtain a
garbled circuit $\bcunblind$, along with encoding functions $\blunblind_x$
for each of the inputs $x$ to the neighbor-computation function in
Figure~\ref{fig:neighbor-computation}.

\item As in the real protocol,
the simulator chooses symmetric encryption keys
$\matrixind{\bar k}{r+1}_{\src,i} \getsr \zo^\enckeylen$ for all $i \in [n]$.
These will be used to encrypt the elements
in the source database on the next round of the protocol.

\item The simulator also chooses four PRF keys
$\bar k_\dne^0, \bar k_\dne^1, \bar k_\dnw^0, \bar k_\dnw^1 \getsr \zo^\prfkeylen$, two
for each axis $\dne, \dnw$. As in the real scheme, the simulator defines the
encryption keys for each direction as follows:
\[
\begin{aligned}
  \bar k_\dn &= F(\bar k_\dne^0, \dn) \oplus F(\bar k_\dnw^0, \dn) \\
  \bar k_\ds &= F(\bar k_\dne^1, \ds) \oplus F(\bar k_\dnw^1, \ds)
\end{aligned} \qquad
\begin{aligned}
  \bar k_\de &= F(\bar k_\dne^0, \de) \oplus F(\bar k_\dnw^1, \de) \\
  \bar k_\dw &= F(\bar k_\dne^1, \dw) \oplus F(\bar k_\dnw^0, \dw).
\end{aligned}
\]

\item The simulator prepares the source database $\overline{\dsrc}$ as follows.
Let $u = v_{r-1}$. If $u \ne \ \perp$, then the
$\ord{u}$ record in $\overline{\dsrc}$ is an encryption under
$\matrixind{\bar k}{r}_{\src, u}$ of the following:
\begin{itemize}
  \item Arithmetic encodings $\blaffine_s = (\blaffine_{\dne, x}, \blaffine_{\dnw, x})$.
  \item Garbled circuit encodings $\blunblind_s(u)$.
  \item Encryptions of the source keys for the neighbors of $u$ in
  the next round of the protocol under the direction keys:
  \[ \ms{Enc}(\bar k_\dn, \matrixind{\bar k}{r+1}_{\src, v_\dn}), \qquad
   \ms{Enc}(\bar k_\de, \matrixind{\bar k}{r+1}_{\src, v_\de}), \qquad
   \ms{Enc}(\bar k_\ds, \matrixind{\bar k}{r+1}_{\src, v_\ds}), \qquad
   \ms{Enc}(\bar k_\dw, \matrixind{\bar k}{r+1}_{\src, v_\dw}), \]
   where $v_\dn, v_\de, v_\dw, v_\dw$ is the neighbor of $u$ in $\mc{G}$ to
   the north, east, south, or west, respectively. If $u$ does not have
   a neighbor in a given direction $j \in \set{\dn, \de, \ds, \dw}$, then
   define $\matrixind{\bar k}{r+1}_{\src, v_j}$ to be the all-zeroes
   string in $\zo^\enckeylen$.
\end{itemize}
For all nodes $u \ne v_{r-1}$, the simulator sets the $\ord{u}$ record
in $\overline{\dsrc}$ to be an encryption of the all-zeroes string
under $\matrixind{\bar k}{r+1}_{\src, u}$.

\item The simulator prepares the destination database $\overline{\ddst}$ as follows.
The $\ord{t}$ record in $\overline{\ddst}$ is an encryption under
$\bar k_{\dst, t}$ of the following:
\begin{itemize}
  \item Arithmetic encodings $\blaffine_t = (\blaffine_{\dne, y}, \blaffine_{\dnw, y})$.
  \item Garbled circuit encodings $\blunblind_t(t)$.
\end{itemize}
For all nodes $u \ne t$, the simulator sets the $\ord{u}$ record in
$\overline{\ddst}$ to be an encryption of the all-zeroes string under
$\bar k_{\dst, u}$.

\item When $\mc{A}$ makes a PIR request for a record in the source database, the simulator plays
the role of the sender in the PIR protocol using $\overline{\dsrc}$ as its input database.

\item When $\mc{A}$ makes a PIR request for a record in the destination database, the simulator plays
the role of the sender in the PIR protocol using $\overline{\ddst}$ as its input database. 

\item When $\mc{A}$ engages in OT for the garbled circuit
encodings of $\hat z_\dne$ and $\hat z_\dnw$,
the simulator replies with the encodings $\blunblind_{z_\dne}(\hat z_\dne)$
and $\blunblind_{z_\dnw}(\hat z_\dnw)$.

\item Let $\dir \in \set{\dn, \de, \ds, \dw}$ be the direction of the edge from
$v_{r - 1}$ to $v_r$ in $\mc{G}$. The simulator sets $\gamma_\dne = 0 = \gamma_\dnw$ and 
$\delta_\dne, \delta_\dnw$ as follows:
\begin{itemize}
  \item If $\hat z_\dne = \bar z_\dne$ and $\hat z_\dnw = \bar z_\dnw$
  and $v_r \ne \ \perp$, then the simulator sets
  \[ \bar \delta_\dne = \begin{cases}
    -1 & \text{if $\dir = \dn$ or $\dir = \de$} \\
     1 & \text{if $\dir = \ds$ or $\dir = \dw$}
  \end{cases} \qquad \text{and} \qquad
  \bar \delta_\dnw = \begin{cases}
    -1 & \text{if $\dir = \dn$ or $\dir = \dw$} \\
     1 & \text{if $\dir = \ds$ or $\dir = \de.$}
  \end{cases} \]

  \item If $\hat z_\dne = \bar z_\dne$ and $\hat z_\dnw = \bar z_\dnw$
  and $v_r = \ \perp$, then the simulator sets
  $\bar \delta_\dne = \xi = \bar \delta_\dnw$. Recall that $\xi \in \F_p$
  satisfies $\xi \notin [-2^\tau, 2^\tau]$.

  \item If exactly one of $\hat z_\dne \ne \bar z_\dne$ or $\hat z_\dnw \ne \bar z_\dnw$
  holds, then with probability $\varepsilon = 2^{\tau+1}/p$, the simulator aborts the
  simulation and outputs $\perp$.
  Otherwise, with probability $1 - \varepsilon$,
  the simulator sets $\bar \delta_\dne = \xi = \bar \delta_\dnw$.

  \item If both
  $\hat z_\dne \ne \bar z_\dne$ and $\hat z_\dnw \ne \bar z_\dnw$, then with
  probability $\varepsilon^2$, the simulator aborts and outputs $\perp$.
  Otherwise, with probability $1 - \varepsilon^2$, the simulator 
  sets $\bar \delta_\dne = \xi = \bar \delta_\dnw$.
\end{itemize}
The simulator sends to $\mc{A}$ the garbled circuit $\bcunblind$, the encodings
of the unblinding coefficients
\[ \blunblind_{\gamma_\dne}(\bar \gamma_\dne), \ 
   \blunblind_{\gamma_\dnw}(\bar \gamma_\dnw), \ 
   \blunblind_{\delta_\dne}(\bar \delta_\dne), \ 
   \blunblind_{\delta_\dnw}(\bar \delta_\dnw), \]
as well as encodings of the PRF keys
\[ \blunblind_{k_\dne^0}(\bar k_\dne^0), \ 
   \blunblind_{k_\dne^1}(\bar k_\dne^1), \ 
   \blunblind_{k_\dnw^0}(\bar k_\dnw^0), \ 
   \blunblind_{k_\dnw^1}(\bar k_\dnw^1). \]

\end{enumerate}
At the end of the protocol execution, the adversary $\mc{A}$ outputs some function
of its view of the execution. If the simulator has not aborted,
the simulator gives the output of $\mc{A}$ to the environment.
This completes the specification of the simulator $\mc{S}$. We now show that
$\mc{S}$ correctly simulates the view of $\mc{A}$ in $\hyb_1$.

\paragraph{Correctness of the simulation.}
First, we define some random variables for the view of the client in the real protocol
and in the simulation. Let $\viewi{0}{\real}$ be the
adversary's view during the setup phase when interacting
according to the real protocol in $\hyb_1$, and let
$\viewi{0}{\mc{S}}$ be the view that $\mc{S}$ simulates for $\mc{A}$ during the setup phase
in the ideal world. In this case,
$\viewi{0}{\real} = (\matrixind{\hat k}{1}_{\src}, \hat k_{\dst})$
where $s$ and $t$ correspond to the source and destination $\mc{A}$
provided as input to the OT protocol. In the simulated view,
$\viewi{0}{\mc{S}} = (\matrixind{\bar k}{1}_{\src,s}, \bar k_{\dst,t})$ with $s, t$
defined similarly.

Next, let $\viewi{r}{\real}$ be the random variable corresponding to the adversary's view
during round $r$ of the real protocol, and let $\viewi{r}{\mc{S}}$ be the view that $\mc{S}$
simulates for $\mc{A}$ during round $r$ in the ideal world. More explicitly, we write
$\viewi{r}{\real} = (\pir_{\src}, \pir_\dst, \tcunblind, \tlunblind)$, where
$\pir_{\src}$ and $\pir_{\dst}$ denote the client's view in the PIR
protocol on databases $\dsrc$ and $\ddst$, respectively, $\tcunblind$
denotes the garbled circuit, and $\tlunblind$
denotes the set of garbled circuit encodings the client receives via the OT protocol
(corresponding to inputs $\hat z_\dne$ and $\hat z_\dnw$) as well as the encodings of the
server's inputs: the unblinding coefficients $\gamma_\dne, \gamma_\dnw, \delta_\dne, \delta_\dnw$,
and the PRF keys $k_\dne^0, k_\dne^0, k_\dnw^0, k_\dnw^1$. Similarly, we define
$\viewi{r}{\mc{S}} = (\overline {\pir_{\src}}, \overline{\pir_\dst}, \bcunblind, \blunblind)$,
where $\overline{\pir_{\src}}$ and $\overline{\pir_{\dst}}$ denote the client's view in the PIR
protocol over databases $\overline{\dsrc}$ and $\overline{\ddst}$, respectively,
and $\blunblind$ denotes the set of garbled circuit encodings the client receives
from the OT protocol as well as the encodings of the server's inputs
$\bar \gamma_\dne, \bar \gamma_\dnw, \bar \delta_\dne, \bar \delta_\dnw$,
$\bar k_\dne^0, \bar k_\dne^1, \bar k_\dnw^0, \bar k_\dnw^1$. Next, define
\[ \viewi{0:r}{\real} = \set{\viewi{i}{\real}}_{i=0}^{r} \]
to be the joint distribution of the view of adversary $\mc{A}$ in the setup and first $r$
rounds of the protocol. We define $\viewi{0:r}{\mc{S}}$ analogously. We now show that
\[ \viewi{0:R}{\real} \appc \viewi{0:R}{\mc{S}}. \]

We first characterize the keys in the simulation. Conceptually, we show
that if a client knows at most one source key in round $r$,
then this property also holds in $r + 1$. This effectively binds the client to
a single consistent path in the course of the protocol execution.
\begin{lemma}
  \label{lem:consistent-paths-sim}
  Let $s, t \in [n]$ be the source and destination nodes the client submits
  to the OT oracle in the setup phase of the simulation. Let $(s = v_0, \ldots, v_R)$ be
  the path the simulator receives from the trusted party. Fix a round $0 < r < R$, and
  suppose the following conditions hold:
  \begin{itemize}
    \item For all nodes $v \ne v_{r-1}$, the conditional distribution of $\matrixind{\bar k}{r}_{\src, v}$
    given $\viewi{0:r-1}{\mc{S}}$ is computationally indistinguishable from the uniform
    distribution on $\zo^\enckeylen$.

    \item For all nodes $v \ne t$, the conditional distribution of $\bar k_{\dst, v}$
    given $\viewi{0:r-1}{\mc{S}}$ is computationally indistinguishable from the uniform distribution
    on $\zo^\enckeylen$.
  \end{itemize}
  Then the corresponding conditions hold for round $r + 1$.
\end{lemma}
\begin{proof}
We first describe the client's view of the protocol execution on the $\ord{r}$
round of the protocol. For notational convenience, we set $u = v_{r-1}$.
We now consider each component in $\viewi{r}{\mc{S}}$ separately:
\begin{itemize}
  \item The client's view of the PIR protocol on the source database. We can
  express the view $\pir_\src$ as a (possibly randomized) function
  $f_1(\viewi{0:r-1}{\mc{S}}, \dsrc)$ in the client's view of the protocol
  execution thus far and the server's database $\overline \dsrc$. Each
  record $v \ne u$ in $\overline \dsrc$ is an encryption of the all-zeroes string.
  As long as $u \ne \ \perp$, the $\ord{u}$ record in $\overline \dsrc$
  contains the following:
  \begin{itemize}
    \item Arithmetic circuit encodings
    $\blaffine_s = (\blaffine_{\dne, x}, \blaffine_{\dnw, x})$
    of the source node $u$.
    
    \item Garbled circuit encodings $\blunblind_s$
    of the node $u$.

    \item Encryptions $\bar \kappa_\dn, \bar \kappa_\de, \bar \kappa_\de, \bar \kappa_\dw$
    of the source keys for the neighbors of $u = v_{r-1}$ in the next round
    of the protocol under the direction keys $\bar k_\dn, \bar k_\de, \bar k_\ds, \bar k_\dw$:
    \[ \bar \kappa_\dn = \ms{Enc}(\bar k_\dn, \matrixind{\bar k}{r+1}_{\src, v_\dn}), \quad
       \bar \kappa_\de = \ms{Enc}(\bar k_\de, \matrixind{\bar k}{r+1}_{\src, v_\de}), \quad
       \bar \kappa_\ds = \ms{Enc}(\bar k_\ds, \matrixind{\bar k}{r+1}_{\src, v_\ds}), \quad 
       \bar \kappa_\dw = \ms{Enc}(\bar k_\dw, \matrixind{\bar k}{r+1}_{\src, v_\dw}), \]
    where $v_\dn, v_\de, v_\ds, v_\dw$ is the neighbor of $u$ in $\mc{G}$ to
    the north, east, south, or west, respectively. If $u$ does not have a
    neighbor in a given direction $\diri \in \set{\dn, \de, \ds, \dw}$, then
    $\matrixind{\bar k}{r+1}_{\src, v_\diri} = \zo^\enckeylen$. 
  \end{itemize}

  \item The client's view of the PIR protocol on the destination database.
  Similar to the previous case, we can express the view $\overline \pir_\dst$ as a
  (possibly randomized) function $f_2(\viewi{0:r-1}{\mc{S}}, \overline{\pir_\src}, \overline \ddst$).
  In hybrid $\hybrid{2}$, every record $v \ne t$ is an encryption of the 
  all-zeroes string. The $\ord{t}$ record consists of the following:
  \begin{itemize}
    \item Arithmetic circuit encodings
    $\blaffine_t = (\blaffine_{\dne, y}, \blaffine_{\dnw, y})$
    of the destination node $t$.

    \item Garbled circuit encodings $\blunblind_t$ of the destination node $t$.
  \end{itemize}

  \item The garbled circuit $\bcunblind$.

  \item The set of garbled circuit encodings from the
  OT protocol. In the OT hybrid model, each OT is replaced
  by an oracle call to the ideal OT functionality.
  Thus, $\blunblind$ consists of a set of garbled circuit encodings
  $\blunblind_{z_\dne}, \blunblind_{z_{\dnw}}$ from the OT protocol, encodings
  $\blunblind_{\gamma_\dne}, \blunblind_{\gamma_\dnw},
   \blunblind_{\delta_\dne}, \blunblind_{\delta_\dnw}$
  of the server's unblinding coefficients and encodings
  $\blunblind_{k_\dne^0}, \blunblind_{k_\dne^1}, \blunblind_{k_\dnw^0}, \blunblind_{k_\dnw^1}$
  of the server's PRF keys.
\end{itemize}
We can express the adversary's view as
\begin{equation}
  \label{eq:view-simulation}
  \viewi{r}{\mc{S}} = \set{ \bar f \left( \viewi{0:r-1}{\mc{S}},
                                         \bar \kappa_\dn, \bar \kappa_\de, \bar \kappa_\ds, \bar \kappa_\dw,
                                         \overline{\affine},
                                         \blunblind_s, \blunblind_t \right),
                                         \bcunblind, \blunblind },
\end{equation}
where $\bar f$ is a (possibly randomized) function, and
$\overline{\ms{affine}} = (\blaffine_s, \blaffine_t)$ are the affine encodings of the source
and destination vectors.

By construction, the set of encodings $\set{\blunblind_s, \blunblind_t, \blunblind}$
constitute a complete set of
encodings for a single unique input $\bar x$ to the neighbor-computation function, and so invoking
Lemma~\ref{lem:yao-simulatable}, we conclude that there exists a $\ppt$
algorithm $\bar{\mc{S}}_1$ such that
\[ \viewi{r}{\mc{S}} \appc \bar{\mc{S}}_1 \left( 
      1^\lambda, \viewi{0:r-1}{\mc{S}},
      \bar \kappa_\dn, \bar \kappa_\de, \bar \kappa_\ds, \bar \kappa_\dw,
      \overline{\affine},
      \cunblind, \cunblind(\bar x).
  \right). \]
To conclude the proof, we condition on the possible outputs of $\cunblind(\bar x)$. There are
two cases:
\begin{itemize}
  \item Suppose $\cunblind(\bar x) = (\hat b_\dne, \hat b_\dnw, \hat k_\dne, \hat k_\dnw)$.
  By definition, this means that $\hat k_\dne = \bar k_\dne^{\hat b_\dne}$ and
  $\hat k_\dnw = \bar k_\dnw^{\hat b_\dnw}$. Let $\dir \in \set{\dn, \de, \ds, \dw}$
  be the direction of the edge from $v_{r-1}$ to $v_r$ in $\mc{G}$. The simulator chooses
  the garbled circuit encodings such that $\dir = \axistodir(\hat b_\dne, \hat b_\dnw)$.
  Thus, by construction of the direction keys $\bar k_\dn, \bar k_\de, \bar k_\ds, \bar k_\dw$,
  and PRF security, we conclude that the conditional distribution of $\bar k_{\dir'}$ given
  $\viewi{r}{\mc{S}}$ is computationally indistinguishable from uniform for all directions
  $\dir' \ne \dir$. Invoking semantic security of $(\ms{Enc}, \ms{Dec})$, we conclude
  that there exists a $\ppt$ algorithm $\bar{\mc{S}}_2$ such that
  \begin{equation}
    \label{eq:simulation-view-char}
    \viewi{r}{\mc{S}} \appc 
      \bar{\mc{S}}_2 \left( 1^\lambda, 
                            \viewi{0:r-1}{\mc{S}},
                            \bar \kappa_\dir,
                            \overline \affine,
                            \cunblind, \cunblind(\bar x) \right).
  \end{equation}
  By definition, $v_r$ is the node in direction $\dir$ with respect to $v_{r-1}$, so 
  \begin{equation}
    \label{eq:sim-direction-key}
    \bar \kappa_\dir = \ms{Enc}(\bar k_\dir, \matrixind{\bar k}{r+1}_{\src, v_\dir})
                      = \ms{Enc}({\bar k_\dir, \matrixind{\bar k}{r+1}_{\src, v_r}}).
  \end{equation}
  Since for all $v \in [n]$, the simulator chooses $\matrixind{\bar k}{r+1}_{\src, v}$
  uniformly and independently (of all other quantities),
  we conclude from the characterization
  in Eq.~\eqref{eq:simulation-view-char} and~\eqref{eq:sim-direction-key}
  that for all nodes $v \ne v_r$, the conditional distribution
  of $\matrixind{\bar k}{r + 1}_{\src, v}$ given $\viewi{0:r}{\mc{S}}$
  is computationally indistinguishable from uniform. The first condition follows.

  \item Suppose $\cunblind(\bar x) = \ \perp$. Using an argument similar to that made for
  the previous case, the conditional distribution
  of the direction keys $\bar k_\dn, \bar k_\de, \bar k_\ds, \bar k_\dw$
  given $\viewi{0:r}{\mc{S}}$
  is computationally indistinguishable from uniform. By semantic
  security of $(\ms{Enc}, \ms{Dec})$, we conclude that the adversary's view
  can be entirely simulated independently of $\matrixind{\bar k}{r+1}_{\src, v}$
  for all $v \in [n]$. Once again, the first condition holds.
\end{itemize}
To conclude the proof, we note that in the characterization
from Eq.~\eqref{eq:view-simulation}, all quantities in $\viewi{r}{\mc{S}}$ are
independent of $\bar k_{\dst, v}$ for $v \ne t$. Thus, if the conditional
distribution of $\bar k_{\dst, v}$ given $\viewi{0:r-1}{\mc{S}}$ is computationally
indistinguishable from the uniform distribution, then the corresponding condition
continues to hold at the end of round $r$.
\end{proof}

\noindent With this preparation, we now argue inductively that
$\viewi{0:R}{\real} \appc \viewi{0:R}{\mc{S}}$. We begin with the base case.

\begin{claim}
\label{claim:security-base-case}
Let $s, t \in [n]$ be the indices of the source and destination nodes,
respectively, that the client submits to the OT oracle in the setup
phase of the protocol. The following conditions hold at the end of the setup
phase of the protocol:
\begin{itemize}
  \item The simulator perfectly simulates the view of the $\mc{A}$:
  $\viewi{0}{\real} \equiv \viewi{0}{\mc{S}}$.

  \item For all nodes $v \ne s$, the conditional
  distributions $\matrixind{k}{1}_{\src, v}$ given $\viewi{0}{\real}$,
  and $\matrixind{\bar k}{1}_{\src, v}$ given $\viewi{0}{\mc{S}}$ are both uniform.

  \item For all nodes $v \ne t$, the conditional distributions
  $k_{\dst, v}$ given $\viewi{0}{\real}$ and $\bar k_{\dst, v}$ given 
  $\viewi{0}{\mc{S}}$ are both uniform.
\end{itemize}
\end{claim}

\begin{proof}
We consider each claim separately:
\begin{itemize}
  \item In the real scheme, $\matrixind{k}{1}_{\src, s}, k_{\dst, t}$ are chosen
  uniformly and independently from $\zo^\enckeylen$. The same is true of
  the keys $\matrixind{\bar k}{1}_{\src, s}, \bar k_{\dst, t}$ in the simulation.
  Since $\viewi{0}{\real} = (\matrixind{k}{1}_{\src, s}, k_{\dst, t})$ and
  $\viewi{0}{\mc{S}} = (\matrixind{\bar k}{1}_{\src, s}, \bar k_{\dst, t})$, we
  conclude that $\viewi{0}{\real} \equiv \viewi{0}{\mc{S}}$.

  \item In the real protocol (resp., the simulation), for all $v \in [n]$, 
  the encryption key $\matrixind{k}{1}_{\src, v}$ (resp., $\matrixind{\bar k}{1}_{\src, v}$)
  is chosen uniformly and independently of all other quantities in the protocol.
  Thus, for $v \ne s$, the distribution of $\matrixind{k}{1}_{\src, v}$ is independent
  of $\viewi{0}{\real} = (\matrixind{k}{1}_{\src, s}, k_{\dst, t})$. Similarly,
  the distribution of $\matrixind{\bar k}{1}_{\src, v}$ is independent of
  $\viewi{0}{\mc{S}}$. The claim follows.

  \item Similar to the previous statement, for all $v \in [n]$, the encryption keys
  $k_{\dst, v}$ and $\bar k_{\dst, v}$ are chosen uniformly and independently of all
  other quantities in the protocol, which proves the claim. \qedhere
\end{itemize}
\end{proof}

\begin{claim}
\label{claim:security-inductive-step}
Fix $0 < r < R$. Suppose the following conditions hold in round $r$:
\begin{itemize}
  \item The view $\viewi{0:r - 1}{\real}$ of the adversary interacting in the real
  protocol is computationally indistinguishable from the view 
  $\viewi{0:r - 1}{\mc{S}}$ of the adversary interacting with the simulator.

  \item For all nodes $v \in [n]$, the conditional distribution of
  $\matrixind{k}{r}_{\src, v}$ given $\viewi{0:r-1}{\real}$ is computationally
  indistinguishable from the uniform distribution over $\zo^\enckeylen$
  if and only if the conditional distribution
  of $\matrixind{\bar k}{r}_{\src, v}$ given $\viewi{0:r-1}{\mc{S}}$ is
  computationally indistinguishable from the uniform distribution over
  $\zo^\enckeylen$.

  \item For all nodes $v \in [n]$, the conditional distribution of
  $k_{\dst, v}$ given $\viewi{0:r-1}{\real}$ is computationally
  indistinguishable from the uniform distribution over $\zo^\enckeylen$
  if and only if the conditional distribution
  of $\bar k_{\dst, v}$ given $\viewi{0:r-1}{\mc{S}}$ is
  computationally indistinguishable from the uniform distribution over
  $\zo^\enckeylen$.
\end{itemize}
Then, the conditions also hold in round $r + 1$.
\end{claim}

\begin{proof}
Let $s, t \in [n]$ be the source and destination nodes the client submits to the
OT oracle in the setup phase of the protocol, and let $(s = v_0, \ldots, v_R)$ be the path
the simulator receives from the trusted party. Consider the view of $\mc{A}$
in the simulation. By Claim~\ref{claim:security-base-case},
for all $v \ne s$, the conditional distribution of the keys $\matrixind{\bar k}{1}_{\src, v}$
given $\viewi{0}{\mc{S}}$ is uniform. Similarly, for all $v \ne t$, the conditional distribution
of the keys $\bar k_{\dst, v}$ given $\viewi{0}{\mc{S}}$ is uniform.
By iteratively applying Lemma~\ref{lem:consistent-paths-sim},
we conclude that for all nodes $v \ne v_{r-1}$, the conditional distribution of
$\matrixind{\bar k}{r}_{\src, v}$ given $\viewi{0:r-1}{\mc{S}}$ 
is computationally indistinguishable from uniform. Similarly, for all
$v \ne t$, the conditional distribution of $\bar k_{\dst, t}$ given
$\viewi{0:r-1}{\mc{S}}$ is computationally indistinguishable from uniform.

Invoking the inductive hypothesis, we have that for all nodes $v \ne v_{r - 1}$,
the conditional distribution of $\matrixind{k}{r}_{\src, v}$ given
$\viewi{0:r-1}{\real}$ is computationally
indistinguishable from uniform. Similarly, for all nodes $v \ne t$,
the conditional distribution of $k_{\dst, v}$ given $\viewi{0:r-1}{\real}$
is computationally indistinguishable from uniform.
We now show that on round $r$,
$(\pir_{\src}, \pir_{\dst}) \appc (\overline{\pir_{\src}}, \overline{\pir_{\dst}})$.
The client's view in the PIR protocol can be regarded as a (possibly randomized) function
of the client's view in the first $r-1$ rounds of the protocol and the server's input database $\mc{D}$.
Since $\viewi{0:r-1}{\real} \appc \viewi{0:r-1}{\mc{S}}$, it suffices to argue 
that the conditional distribution of $(\dsrc, \ddst)$ given $\viewi{0:r-1}{\real}$
is computationally indistinguishable
from the conditional distribution of $(\overline{\dsrc}, \overline{\ddst})$ given
$\viewi{0:r-1}{\mc{S}}$.

For notational convenience, let $u = v_{r-1}$. We use a hybrid argument.
We define the following hybrid experiments:
\begin{itemize}
  \item Hybrid $\hybrid{0}$ is the real game where the server prepares $\dsrc$
  and $\ddst$ as described in Figure~\ref{fig:protocol}.

  \item Hybrid $\hybrid{1}$ is identical to $\hybrid{0}$, except
  the server substitutes an encryption of the all-zeroes string under
  $\matrixind{k}{r}_{\src, v}$ for all records
  $v \ne u$ in $\dsrc$. 

  \item Hybrid $\hybrid{2}$ is identical to $\hybrid{1}$,
  except the server substitutes an encryption of the
  all-zeroes string under $k_{\dst, v}$ for all records
  $v \ne t$ in $\ddst$.
\end{itemize}
Since for all $v \ne u$, the conditional distribution of $\matrixind{k}{r}_{\src, v}$
given $\viewi{0:r-1}{\real}$ is computationally indistinguishable from uniform,
we can appeal to the semantic security of $(\ms{Enc}, \ms{Dec})$ to conclude
that hybrid experiments $\hybrid{0}$ and $\hybrid{1}$ are computationally
indistinguishable. Similarly, since for all $v \ne t$, the conditional distribution of
$k_{\dst, v}$ given $\viewi{0:r-1}{\real}$ is computationally indistinguishable from
uniform, we have that $\hybrid{1}$ and $\hybrid{2}$ are computationally
indistinguishable.

We now show that the joint distribution of
$(\dsrc, \ddst)$ in $\hybrid{2}$ is computationally indistinguishable from
the joint distribution of $(\overline \dsrc, \overline \ddst)$ in the simulation.
In $\hybrid{2}$,
every record $v \ne u$ in $\dsrc$ is an encryption of the all-zeroes string
under a key $\matrixind{k}{r}_{\src, v}$
that is computationally indistinguishable from uniform given the
adversary's view of the protocol thus far; the same is true in the
simulation. Similarly, every record $v \ne t$ in
$\ddst$ is an encryption of the all-zeroes string under a key
$k_{\dst, v}$ that looks computationally indistinguishable from uniform
to the adversary. This is the case in the simulation.
Let $r_{\src, u}$ be the $\ord{u}$ record in $\dsrc$ and let
$r_{\dst, t}$ be the $\ord{t}$ record in $\ddst$. Similarly, let
$\bar r_{\src, u}$ be the $\ord{u}$ record in $\overline \dsrc$ and let
$\bar r_{\dst, t}$ be the $\ord{t}$ record in $\overline \ddst$. It suffices now
to show that $(r_{\src, u}, r_{\dst, t})$ is computationally indistinguishable
from $(\bar r_{\src, u}, \bar r_{\dst, t})$.
In the real scheme, the record $r_{\src, u}$ contains the following components:
\begin{itemize}
    \item Arithmetic circuit encodings
    $\tlaffine_{\dne, x}(\matrixind{A}{\dne}_u), \tlaffine_{\dnw, x}(\matrixind{A}{\dnw}_u)$
    for the source $u$.
    
    \item Garbled circuit encodings $\tlunblind_s(u)$
    for the source $u$.

    \item Encryptions $\kappa_\dn, \kappa_\de, \kappa_\ds, \kappa_\dw$
    of the source keys for the neighbors of $u$ in the next round
    of the protocol under the direction keys $k_\dn, k_\de, k_\ds, k_\dw$:
    \[ \kappa_\dn = \ms{Enc}(k_\dn, \matrixind{k}{r+1}_{\src, v_\dn}), \quad
       \kappa_\de = \ms{Enc}(k_\de, \matrixind{k}{r+1}_{\src, v_\de}), \quad
       \kappa_\ds = \ms{Enc}(k_\ds, \matrixind{k}{r+1}_{\src, v_\ds}), \quad 
       \kappa_\dw = \ms{Enc}(k_\dw, \matrixind{k}{r+1}_{\src, v_\dw}), \]
    where $v_\dn, v_\de, v_\ds, v_\dw$ is the neighbor of $u$ in $\mc{G}$ to
    the north, east, south, or west, respectively. If $u$ does not have a
    neighbor in a given direction $\diri \in \set{\dn, \de, \ds, \dw}$, then
    $\matrixind{k}{r+1}_{\src, v_\diri} = \zo^\enckeylen$.
\end{itemize}
The record $r_{\dst, t}$ contains the following components:
\begin{itemize}
  \item Arithmetic circuit encodings
  $\tlaffine_{\dne, y}(\matrixind{B}{\dne}_t), \tlaffine_{\dnw, y}(\matrixind{B}{\dnw}_t)$
  for the destination $t$.

  \item Garbled circuit encodings $\tlunblind_t(t)$ for the destination $t$.
\end{itemize}
In the real protocol, the arithmetic circuit encodings are constructed
independently of the
garbled circuit for the neighbor-computation function. The neighbor keys
$k_\diri$ for $\diri \in \set{\dn, \de, \ds, \dw}$
and source keys $\matrixind{k}{r+1}_{\src, v_\diri}$ for the subsequent round of the
protocol are also generated independently of both the arithmetic circuit
encodings and the garbled circuit. Thus, the joint distribution decomposes into
three product distributions over the arithmetic circuit encodings,
the garbled circuit encodings, and the encryptions of the source keys for the
next round. We reason about each distribution separately:
\begin{itemize}
  \item The simulator constructs the garbled circuit for the neighbor-computation
  function exactly as in the real scheme. Thus, the garbled circuit encodings
  in $r_{\src, u}, \bar r_{\src, u}$ and 
  $r_{\dst, t}, \bar r_{\dst, t}$ are identically distributed.

  \item The neighbor keys $\bar k_\diri$ for $\dir \in \set{\dn, \de, \ds, \dw}$
  in the simulation are generated exactly as the keys $k_\diri$ in the real scheme.

  \item In the real scheme, the affine encodings $\tlaffine_{\dne, x}(\matrixind{A}{\dne}_u)$
  and $\tlaffine_{\dne, y}(\matrixind{B}{\dne}_t)$ evaluate to 
  \[ z_\dne = \alpha_\dne \langle \matrixind{A}{\dne}_u, \matrixind{B}{\dne}_t \rangle + \beta_\dne, \]
  where $\alpha_\dne$ is uniform in $\F_p^*$ and $\beta_\dne$ is uniform in $\F_p$.
  In particular, this means that $z_\dne$ is distributed uniformly over $\F_p$. This is
  precisely the same distribution from which the simulator samples $\bar z_\dne$.
  Together with Lemma~\ref{lem:affine-encoding-simulatable}, we conclude that
  \[ \mc{S}_{\ms{ac}}(\bar z_\dne) \equiv 
     \ms{S}_{\ms{ac}}(z_\dne) \equiv 
     (\tlaffine_{\dne, x}(\matrixind{A}{\dne}_u), \tlaffine_{\dne, y}(\matrixind{B}{\dne}_t)). \]
  An analogous argument shows that
  \[ \mc{S}_{\ms{ac}}(\bar z_\dnw) \equiv 
     \ms{S}_{\ms{ac}}(z_\dnw) \equiv 
     (\tlaffine_{\dnw, x}(\matrixind{A}{\dnw}_u), \tlaffine_{\dnw, y}(\matrixind{B}{\dnw}_t)), \]
  where $z_\dnw = \alpha_\dnw \langle \matrixind{A}{\dnw}_u, \matrixind{B}{\dnw}_t \rangle + \beta_\dnw$.
  We conclude that the arithmetic circuit encodings are identically distributed in
  both the real scheme and the simulation.
\end{itemize}
We conclude
that $(r_{\src, u}, r_{\dst, t}) \equiv (\bar r_{\src, u}, \bar r_{\dst, t})$,
and correspondingly,
$(\dsrc, \ddst) \appc (\overline{\dsrc}, \overline{\ddst})$.
In the real protocol, the client's view $\pir_\src, \pir_\dst$ in the PIR
protocols can be expressed as an efficiently-computable and possibly randomized
function $f$ of its view $\viewi{0:r-1}{\real}$ in the first $r-1$ rounds 
of the protocol and the server's databases $\dsrc, \ddst$:
\[ (\pir_\src, \pir_\dst) \equiv f(\viewi{0:r-1}{\real}, r_{\src,u}, r_{\dst, t}). \]
In the simulation, the simulator synthesizes databases $\overline{\dsrc}, \overline{\ddst}$,
and then plays the role of the server in the PIR protocol. Thus, we have that
\[ (\overline{\pir_\src}, \overline{\pir_\dst}) \equiv
     f(\viewi{0:r-1}{\mc{S}}, \bar r_{\src,u}, \bar r_{\dst, t}). \]
Moreover, we note that the records $r_{\src, u}, r_{\dst, t}, \bar r_{\src, u}, \bar r_{\dst, t}$
are constructed independently of variables from the previous round of the protocol.
Thus, using the inductive hypothesis, $\viewi{0:r-1}{\real} \appc \viewi{0:r-1}{\mc{S}}$,
and the fact that $(r_{\src, u}, r_{\dst, t}) \equiv (\bar r_{\src, u}, \bar r_{\dst, t})$,
we conclude that
\begin{equation}
  \label{eq:pir-views-indis}
  (\pir_{\src}, \pir_{\dst}) \equiv 
     f(\viewi{0:r-1}{\real}, r_{\src,u}, r_{\dst, t}) \appc
     f(\viewi{0:r-1}{\mc{S}}, \bar r_{\src,u}, \bar r_{\dst, t}) \equiv
     (\overline{\pir_{\src}}, \overline{\pir_{\dst}}).
\end{equation}
Using this characterization, we can write
\begin{equation}
  \label{eq:real-view}
  \viewi{0:r}{\real} = \set{\viewi{0:r-1}{\real}, \,
   f(\viewi{0:r-1}{\real}, \, r_{\src,u}, r_{\dst, t}), \,
   \tcunblind, \tlunblind},
\end{equation}
and similarly,
\begin{equation}
  \label{eq:simulation-view}
  \viewi{0:r}{\mc{S}} = \set{\viewi{0:r-1}{\mc{S}}, \,
    f(\viewi{0:r-1}{\mc{S}}, \, \bar r_{\src,u}, \bar r_{\dst, t}), \,
    \bcunblind, \blunblind}.
\end{equation}
From this, we express $\viewi{0:r}{\real}$ as an efficiently-computable
and possibly randomized function $f'$ in the garbled circuit components and the auxiliary
components:
\begin{equation}
 \label{eq:view-real-gc}
 \viewi{r}{\real} \equiv f'\left( \tlunblind_s(u), \, \tlunblind_t(t), \, \tlunblind, \,
   \tcunblind, \, \ms{aux} \right),
\end{equation}
where $\ms{aux}$ contains the additional variables $\viewi{r}{\real}$ depends on
that are independent of the garbled circuit encodings: $\viewi{0:r-1}{\real}$, the
arithmetic circuit encodings, and the encryptions of the source keys for the next
round of the protocol.

Since the simulator prepares the garbled circuit exactly as in the real protocol,
from the characterization of $\viewi{0:r}{\mc{S}}$ in Eq.~\eqref{eq:simulation-view},
we can similarly write
\begin{equation}
 \label{eq:view-sim-gc}
 \viewi{r}{\mc{S}} \equiv f'\left( \blunblind_s(u), \, \blunblind_t(t), \, \blunblind, \,
   \bcunblind, \, \overline{\ms{aux}} \right),
\end{equation}
where $\overline{\ms{aux}}$ contains the same auxiliary variables as $\ms{aux}$.
From Eq.~\eqref{eq:pir-views-indis}, we have in particular that
$\ms{aux} \appc \overline{\ms{aux}}$.

By construction, the encodings $\tlunblind_s(u), \, \tlunblind_t(t), \, \tlunblind$
in the real scheme constitutes a complete set of encodings for the garbled circuit
$\tcunblind$. Let $\tilde x$ be the associated input to the neighbor-computation
circuit $\cunblind$. In particular,
$\tilde x = (u, t, \hat z_\dne, \hat z_\dnw, \gamma_\dne, \gamma_\dnw,
             \delta_\dne, \delta_\dnw, k_\dne^0, k_\dne^1, k_\dnw^0, k_\dnw^1)$,
where $\hat z_\dne$ and $\hat z_\dnw$ are the encodings the client OTs for in the
protocol execution. Similarly, the encodings
$\blunblind_s(u), \, \blunblind_t(t), \, \blunblind$ constitutes a complete set of 
encodings for the garbled circuit $\bcunblind$. Let
$\bar x = (u, t, \hat z_\dne, \hat z_\dnw, \bar \gamma_\dne, \bar \gamma_\dnw,
           \bar \delta_\dne, \bar \delta_\dnw, \bar k_\dne^0, \bar k_\dne^1,
           \bar k_\dnw^0, \bar k_\dnw^1)$ be the associated input
to $\cunblind$ in the simulation. Moreover, by the characterization of the client's
view in Eq.~\eqref{eq:real-view} and~\eqref{eq:simulation-view}, on each
round of the protocol execution, we can associate two unique sets of affine encodings
with the client's view, one for each axis. Let $z_\dne$ and $z_\dnw$ denote the values
to which these two sets of affine encodings evaluate. In the real execution, we thus have
$z_\dne = \alpha_\dne \langle \matrixind{A}{\dne}_u, \matrixind{B}{\dne}_t \rangle + \beta_\dne$
and $z_\dnw = \alpha_\dnw \langle \matrixind{A}{\dnw}_u, \matrixind{B}{\dnw}_t \rangle + \beta_\dnw$.
In the simulation, we have $z_\dne = \bar z_\dne$ and $z_\dnw = \bar z_\dnw$.
We now consider three cases:

\begin{itemize}
  \item Suppose that the client OTs for the encodings of inputs consistent with
  the outputs of the arithmetic circuit. In other words, $\hat z_\dne = z_\dne$
  and $\hat z_\dnw = z_\dnw$. By definition, if
  $u = t$, then $\cunblind(\tilde x) = \ \perp$. Otherwise, by correctness of the
  affine encodings,
  $\cunblind(\tilde x) = (\hat b_\dne, \hat b_\dnw, k_\dne^{\hat b_\dne}, k_\dnw^{\hat b_\dnw})$,
  where $\dir = \axistodir(\hat b_\dne, \hat b_\dnw)$ is the direction of travel from
  $u$ to $t$, as determined by the next-hop routing matrices
  $\matrixind{A}{\dne}, \matrixind{B}{\dne}, \matrixind{A}{\dnw}, \matrixind{B}{\dnw}$.
  In particular, $\dir$ is the direction of the edge from $u = v_{r-1}$ to $v_r$.
  In the simulation, when $\hat z_\dne = \bar z_\dne$ and
  $\hat z_\dnw = \bar z_\dnw$, the simulator chooses unblinding factors
  $\bar \gamma_\dne, \bar \gamma_\dnw, \bar \delta_\dne, \bar \delta_\dnw$
  such that $\cunblind(\bar x) = \cunblind(\tilde x)$. Now, invoking
  the input-privacy of the garbling scheme (Definition~\ref{def:garbling-input-privacy}),
  we conclude that
  \begin{align*}
    \set{ \tlunblind_s(u), \tlunblind_t(t), \tlunblind, \tcunblind}
     & \appc \mc{S}_\yao(1^\lambda, \cunblind, \cunblind(\bar x)) \\
     & \appc \set{ \blunblind_s(u), \blunblind_t(t), \blunblind, \bcunblind}.
  \end{align*}
  Since the garbled circuit components of $\viewi{r}{\real}$ are computationally
  indistinguishable from $\viewi{r}{\mc{S}}$, we conclude that
  $\viewi{r}{\real} \appc \viewi{r}{\mc{S}}$.
  \\\\
  To see the second condition of Claim~\ref{claim:security-inductive-step} holds,
  we appeal to input privacy of the garbling scheme to argue that the conditional
  distribution of the keys $k_\dne^{1-\hat b_\dne}$ and $k_\dnw^{1-\hat b_\dnw}$
  is computationally indistinguishable from uniform given $\viewi{0:r}{\real}$.
  As in the proof of Lemma~\ref{lem:consistent-paths-sim}, the conditional
  distribution of $k_\dir'$ for all directions $\dir' \ne \dir \in \set{\dn, \de, \ds, \dw}$
  is computationally indistinguishable from uniform given $\viewi{0:r}{\real}$. Finally,
  invoking semantic security of $(\ms{Enc}, \ms{Dec})$, we conclude that
  for all $v \ne v_r$, the conditional distribution of $\matrixind{k}{r+1}_{\src, v}$
  is computationally indistinguishable from uniform. As shown in the proof of
  Lemma~\ref{lem:consistent-paths-sim}, this is precisely the
  case in the simulation.
  \\\\
  The third condition of Claim~\ref{claim:security-inductive-step} holds
  since the components in the client's view in round $r$ of the real scheme
  (Eq.~\ref{eq:real-view}) are independent of the destination keys
  $k_{\dst, v}$ for all $v \ne t$. Since the conditional distribution of $k_{\dst, v}$
  for $v \ne t$ given $\viewi{0:r-1}{\real}$ is computationally indistinguishable from
  uniform, the conditional distribution remains uniform conditioned on
  $\viewi{0:r}{\real}$. This precisely matches the distribution in the simulation
  (Lemma~\ref{lem:consistent-paths-sim}).
 
  \item Suppose that $u = t$. Then, $\cunblind(\tilde x) = \ \perp \ = \cunblind(\bar x)$.
    As in the previous case, input privacy of the garbling scheme yields
    $\viewi{r}{\real} \appc \viewi{r}{\mc{S}}$.
    \\\\
    It is not difficult to see that the second condition of
    Claim~\ref{claim:security-inductive-step} holds. Since $\cunblind(\tilde x) = \ \perp$,
    the conditional distribution of the keys $k_\dne^0, k_\dne^1, k_\dnw^0, k_\dnw^1$ is
    computationally indistinguishable from uniform given $\viewi{0:r}{\real}$. By semantic security
    of $(\ms{Enc}, \ms{Dec})$, we conclude that the conditional distribution of
    $\matrixind{k}{r+1}_{\src, v}$ given $\viewi{0:r}{\real}$ is computationally indistinguishable
    from uniform for all $v \in [n]$. By the case analysis in the proof of
    Lemma~\ref{lem:consistent-paths-sim}, this is also the case in the simulation.
    \\\\
    The third condition follows as in the previous case.

  \item Suppose that $u \ne t$ and moreover,
  the client OTs for encodings of inputs that are inconsistent
  with the outputs of the arithmetic circuit. We consider two possibilities.
  \begin{itemize}
    \item Suppose that exactly one of $\hat z_\dne \ne z_\dne$ and $\hat z_\dnw \ne z_\dnw$
    hold. Without loss of generality, suppose that $\hat z_\dne \ne z_\dne$.
    In the real scheme, this means that
    $\hat z_\dne \ne \alpha_\dne \langle \matrixind{A}{\dne}_u, \matrixind{B}{\dne}_t \rangle + \beta_\dne$.
    Next, since $\gamma_\dne = \alpha_\dne^{-1}$ and
    $\delta_\dne = \alpha_\dne^{-1}\beta_\dne$,
    where $\alpha_\dne$ is uniform in $\F_p^*$ and $\beta_\dne$ is uniform in $\F_p$. Thus,
    $\gamma_\dne$ and $\delta_\dne$ are also distributed uniformly over
    $\F_p^*$ and $\F_p$, respectively. Since the family of functions
    $\set{h_{\gamma, \delta}(z) = 
    \gamma z + \delta \pmod p \mid \gamma \in \F_p^*, \delta \in \F_p}$ is pairwise independent,
    it follows that for any distinct $z_\dne, z_\dne^\star \in \F_p$, and all $a, b \in \F_p$,
    \[ \Pr \left[ h_{\gamma_\dne, \delta_\dne}(z_\dne) = a 
      \wedge h_{\gamma_\dne, \delta_\dne}(z_\dne^\star) = b \right] = \frac{1}{p^2}, \]
    where the probability is taken over the randomness in $\gamma_\dne$ and $\delta_\dne$.
    The client's choice of $\hat z_\dne$ and $\hat z_\dnw$ depends only on its view
    $\viewi{0:r-1}{\real}$ of the protocol execution in the previous rounds
    of the protocol, as well as its view
    $\pir_\src$ and $\pir_\dst$ in the PIR protocol. The quantities $\gamma_\dne$ and
    $\delta_\dne$ are sampled independently of $\viewi{0:r-1}{\real}$.
    By the above characterization of $\pir_\src$ and $\pir_\dst$,
    the joint distribution of $(\pir_\src, \pir_\dst)$ is entirely simulatable given
    only $z_\dne$ and variables that are independent of $\gamma_\dne$ and
    $\delta_\dne$ (by invoking the simulator for the affine encodings). Thus, by
    pairwise independence, we conclude that
    \[ \Pr \left[ h_{\gamma_\dne, \delta_\dne}(\hat z_\dne) \in [-2^\tau, 2^\tau] \right] =
       \frac{2^{\tau+1}}{p} = \varepsilon. \]
    If $\hat z_\dne \ne z_\dne$, then with probability $1 - \varepsilon$,
    $\cunblind(\tilde x) = \ \perp$. With probabilty $\varepsilon$,
    $\cunblind(\tilde x) \ne \ \perp$, but this precisely corresponds to
    the first abort condition
    in experiment $\hyb_1$. Thus, in $\hyb_1$, with probability
    $1 - \varepsilon$, $\cunblind(\tilde x) = \ \perp$ and with probability $\varepsilon$,
    the protocol execution aborts.
    \\\\
    In the simulation, the simulator chooses the unblinding factors
    $\bar \gamma_\dne, \bar \gamma_\dnw, \bar \delta_\dne, \bar \delta_\dnw$,
    such that $\cunblind(\bar x) = \ \perp$ with probability
    $1 - \varepsilon$. With probability $\varepsilon$, the simulator aborts.
    We conclude by input privacy of the garbling scheme that the simulation
    is correct. The analysis for the case where $\hat z_\dnw \ne z_\dnw$, but
    $\hat z_\dne = z_\dne$ is entirely analogous.

    \item Suppose that both $\hat z_\dne \ne z_\dne$ and $\hat z_\dnw \ne z_\dnw$.
    By the same analysis as in the first case, we have that
    \[ \Pr \left[ h_{\gamma_\dne, \delta_\dne}(\hat z_\dne) \in [-2^\tau, 2^\tau] \right] =
       \varepsilon = 
       \Pr \left[ h_{\gamma_\dnw, \delta_\dnw}(\hat z_\dnw) \in [-2^\tau, 2^\tau] \right]. \]
    Since the two events are independent, $\cunblind(\tilde x) \ne \ \perp$
    with probability $\varepsilon^2$, and the experiment aborts in $\hyb_1$. With probability
    $1 - \varepsilon^2$, $\cunblind(\tilde x) = \ \perp$. In the simulation, the simulator
    chooses the unblinding factors such that $\cunblind(\bar x) = \ \perp$ with
    probability $1 - \varepsilon^2$, and aborts with probability $\varepsilon^2$. Correctness
    of the simulation follows by input privacy of the garbling scheme.
  \end{itemize}
  Since $\cunblind(\tilde x)$ is either equal to $\perp$ or the experiment aborts, 
  the proof of the second and third statements of Claim~\ref{claim:security-inductive-step}
  follows exactly as in the previous case.
\end{itemize}

Combining Claim~\ref{claim:security-base-case} and
Claim~\ref{claim:security-inductive-step}, we conclude by induction on $r$ that
$\viewi{0:R}{\real} \appc \viewi{0:R}{\mc{S}}$. Thus, the view of the protocol
execution simulated by $\mc{S}$ for $\mc{A}$ is computationally indistinguishable
from the view of $\mc{A}$ interacting with the server in $\hyb_1$.
Correctness of the simulation follows.
\end{proof}

\paragraph{Proof of Claim~\ref{claim:hyb-0-1}.} To conclude the proof of
Theorem~\ref{thm:malicious-client-security}, we show Claim~\ref{claim:hyb-0-1},
or equivalently, that no efficient adversary is able to 
distinguish $\hyb_0$ from $\hyb_1$ except with advantage
$\negl(\lambda) + R \cdot 2^{-\mu}$. Let $\mc{A}$ be
a distinguisher between $\hyb_0$ and $\hyb_1$ with distinguishing
advantage $\adv$:
\[ \adv = \abs{\Pr\left[ \hyb_0(\lambda, \pi, \mc{A}, \mc{E}) \right] - 
          \Pr\left[ \hyb_1(\lambda, \pi, \mc{A}, \mc{E}) \right]}. \]
By construction, $\hyb_0$ and $\hyb_1$
are identical experiments, except experiment $\hyb_1$ terminates if the
abort event in Figure~\ref{fig:hyb-1-bad-event} occurs. Thus, it must
be the case that $\mc{A}$ is able to cause the bad event to occur
with probability $\adv$ in the real experiment. But by
Claim~\ref{claim:hyb-1-2}, the real protocol execution experiment with
the abort event is computationally indistinguishable from the ideal-world
execution with the simulator $\mc{S}$ described in the proof of
Claim~\ref{claim:hyb-1-2}. Thus, if $\mc{A}$ is able to trigger the abort
event in the real protocol execution with probability $\adv$, it is able to
trigger the abort event when interacting with the simulator $\mc{S}$
with probability that is negligibly close to $\adv$. It suffices now to bound
the probability that the simulator $\mc{S}$ aborts the protocol execution.
On each round $r$, the simulator $\mc{S}$ aborts with probability at most
$\varepsilon = 2^{\tau + 1} / p \le 2^{-\mu}$ since
$p > 2^{\mu + \tau + 1}$ (irrespective of the computational power of
the adversary). By a union bound over the total number of rounds $R$,
we conclude that $\mc{S}$
aborts with probability at most $R \cdot 2^{-\mu}$ in the protocol execution.
Thus, the probability that $\mc{A}$ can trigger the abort event in the 
real protocol must be negligibly close to $R \cdot 2^{-\mu}$.
We conclude that $\adv \le \negl(\lambda) + R \cdot 2^{-\mu}$, which proves
the claim. \qed

Claims~\ref{claim:hyb-0-1} and~\ref{claim:hyb-1-2} show that no
no efficient environment can distinguish the real-world execution ($\hyb_0$)
from the ideal-world execution ($\hyb_2$), except with advantage negligibly
close to $R \cdot 2^{\mu}$. This proves
Theorem~\ref{thm:malicious-client-security}. \qed

}{}

\end{document}